\documentclass[11pt,a4paper]{article}
\renewcommand{\arraystretch}{1.5}
\usepackage{amssymb,amsbsy,amsmath,color,comment,epsfig,harvard,rotating}
\usepackage{adjustbox,amsthm,mathtools,bbm}
\usepackage{dsfont,hyperref,pifont,ulem}
\usepackage{subcaption}
\usepackage{settings}
\usepackage{makecell}
\usepackage{booktabs}
\usepackage{multirow}

\newcommand{\xp}{{\textrm{XP}}_{\!\tau}}
\newcommand{\xpb}[1]{{\widehat{\textrm{XP}}_{\!\tau,#1}}}

\begin{document}
\bibliographystyle{joe}
\citationstyle{dcu}
\lineskip=1ex \baselineskip 4ex
\setlength{\arraycolsep}{0.02in}
\pagestyle{empty}
\begin{center}
\noindent{\Large\bf Estimation risk in conditional expectiles}\vskip 3em
\end{center}
\textbf{Marcelo Fernandes}\\
Sao Paulo School of Economics, FGV\\\\
\textbf{Víctor Henriques}\\
Argus Media\\\\
\textbf{Eduardo Fonseca Mendes}\\
Sao Paulo School of Business Administation, FGV\vskip 3em
\noindent\textbf{Abstract:}~~We establish the consistency and asymptotic normality of a two-step estimator of conditional expectiles in the context of conditional scale models. We first estimate the conditional variance parameters by quasi-maximum likelihood and then compute the unconditional expectile of the innovations using the empirical distribution of the standardized residuals. We show how replacing true innovations with standardized residuals affects the asymptotic variances of both conditional and unconditional expectile estimators. Finally, our empirical analysis reveals that conditional expectiles assess tail risk in cryptomarkets in a more robust manner than traditional quantile-based risk measures, such as value at risk and expected shortfall.\vfill

{\small\noindent\textbf{Acknowledgments:}~~We are indebted to Silvia Gonçalves, Giuseppe Cavaliere, Abdelaati Daouia, Alan De Genaro, and Edu Horta, as well as to seminar participants at the UBRI Connect (Zurich, 2024), High-Voltage Econometrics Workshop (Lecce, 2024), Research in Options (Rio de Janeiro, 2024), São Paulo School of Advanced Sciences on High-Dimensional Modeling (São Paulo, 2025), BCN Workshop in Financial Econometrics (Barcelona, 2025), Italian Congress of Econometrics and Empirical Economics (Palermo, 2025), SoFiE Annual Conference (Paris, 2025), TSE Financial Econometrics Conference (Hammamet, 2025), Brazilian School of Time Series and Econometrics (Campinas, 2025), Meetings of the Brazilian Finance Society (São Paulo, 2025), Meetings of the Brazilian Econometric Society (São Paulo, 2025), Workshop on Econometrics and Complex Networks (Buenos Aires, 2026), Virtual Time Series Seminar, University of Illinois at Urbana-Champaign, and University of São Paulo. Fernandes and Mendes acknowledge financial support not only from FAPESP (2023/01728-0) and CNPq (302278/2018-4 and 302153/2022-5), but also from the Silicon Valley Community Foundation through the University Blockchain Research Initiative (UBRI, 2022/199610). The usual disclaimer applies.}
	
\newpage\pagestyle{plain}
\renewcommand{\arraystretch}{2}
\lineskip=1ex \baselineskip 5ex
\section{Introduction}

Quantile-based measures, such as value at risk (VaR) and expected shortfall (ES), dominate risk assessments. The Basel II regulatory framework calculates minimum capital requirements using a VaR approach (\url{www.bis.org/publ/bcbs24.pdf}). However, value at risk is not a satisfactory risk measure for two reasons. First, it does not adequately account for diversification gains because it violates the subadditivity property that characterizes coherent risk measures  \cite{artzner1999coherent,follmer2002convex}. Second, since the value at risk corresponds to a quantile of the distribution, it defines what constitutes a tail realization, but without assessing expected severity. Accordingly, the Basel III regulatory framework suggests using expected shortfall (\url{www.bis.org/publ/bcbs265.pdf}). By measuring the expected loss in excess of the value at risk, ES effectively accounts for not only diversification gains, but also the severity of the tail realization.

The quest for better risk measures is not over yet. \cn{gneiting2011} demonstrates that ES is not individually elicitable, so there is no natural backtest procedure to assess the performance of ES forecasts. In addition, quantile-based risk measure estimates consider only the relative frequency of observations above or below their corresponding predictions \cite{daouia2018estimation}, depending heavily on the tail of the loss distribution \cite{kuan2009assessing,daouia2019extremiles}. More specifically, ES is too conservative because it restricts attention to a given quantile, whereas VaR is too lenient because it ignores severity. As such, risk measures based on a given quantile arguably either underestimate or overestimate risk exposures. 

The alternative class of risk measures based on \pc{newey1987asymmetric} asymmetric least squares (ALS) has been gaining traction in the literature. In particular, expectiles (XP) are a least-squares analog of quantiles, offering the only  coherent (law-invariant) risk measure that simultaneously accounts for diversification gains and allows straightforward backtesting \cite{ziegel2016coherence,daouia2023expectile}. Moreover, it depends on both the probability and severity of the tail event, making them particularly suitable for actuarial and portfolio-allocation problems. This paper establishes the consistency and asymptotic normality of the two-step estimators of the conditional and unconditional expectiles for conditional scale models. As in \cn{francq2015risk}, we first estimate the scale parameters by Gaussian quasi-maximum likelihood (QML) and then estimate the innovation expectile using the empirical distribution of the standardized residuals. Because innovations are unobservable, inference must account for the first-step estimation error. We derive a Bahadur representation that explicitly shows how replacing innovations with standardized residuals affects the second-step estimation. We also obtain a closed-form asymptotic variance, consistent feasible variance estimators, and a conditional Gaussian approximation for the conditional expectile given past realizations.

We corroborate our asymptotic results with Monte Carlo experiments based on designs that reproduce the main stylized facts of asset returns. The simulations assess the finite-sample bias and dispersion of the two-step estimators of the conditional and unconditional expectiles, as well as the quality of the Gaussian approximation based on the asymptotic variance estimator. Both bias and root mean squared error decrease sharply with sample size, whereas standard errors are on average very close to the standard deviation over Monte Carlo replications. In addition, the actual coverage rates of the Wald confidence intervals for both conditional and unconditional expectiles are close to nominal values.

Our approach aligns with the literature. \cn{gao2008estimation} establish consistency and asymptotic normality of the value-at-risk and expected shortfall based on GARCH standardized residuals. \cn{francq2015risk} extend the asymptotic theory to cover several GARCH-type specifications: e.g., exponential GARCH, asymmetric power ARCH, and GJR-GARCH models \cite{nelson1991conditional,ding1993long,glosten1993relation}. We are the first to assess how the estimation of conditional scale parameters affects the inference on conditional and unconditional expectiles. \cn{holzmann2016expectile} and \cn{kratschmer2017statistical} develop the asymptotic theory for the estimation of unconditional expectiles at a fixed level $\tau$ using independent and identically distributed (iid) data. \cnm{daouia2018estimation} (2018, \cy{daouia2020tail}) derive the asymptotic distribution of a weighted ALS estimator for extreme expectiles at level $\tau_n$, with $\tau_n\to 1$ as the sample size $n$ grows. \cn{girard2021extreme} extend the analysis to consider the estimation of conditional extreme expectiles in heavy-tailed heteroskedastic regressions using a two-step approach. In particular, they show that substituting standardized residuals for true innovations does not affect the asymptotic distribution of extreme quantile estimators because the latter converges at a slower rate $\sqrt{n(1-\tau_n)}$ than the first-step estimation error shrinks to zero. Unfortunately, this is not the case here. The estimation of expectiles at a fixed level $\tau$ converges at the same $\sqrt{n}$-rate as the estimation of the conditional mean and variance parameters, and hence the latter affects the asymptotic distribution of the former.

Finally, we empirically assess the performance of conditional expectiles relative to quantile-based risk measures in cryptocurrency markets \ca{MS22}{for an overview, see}. Our motivation is twofold. First, crypto assets appeal to investors mainly due to their low correlation with traditional asset classes \ca{BB22}{see, among others,}. This suggests that we should not assess tail risk by looking only at the value at risk, as we would miss out any diversification benefit. Second, extreme tail events are relatively common in crypto markets \cite{gkillas2018application,STT18,borri2019,nguyen2020investigating}. This casts doubt on the suitability of the expected-shortfall measure, given that the latter is very hard to estimate precisely under heavy tails. As such, crypto markets make fertile ground for the application of conditional expectiles in the stipulation of minimum capital requirements.

We find that the one-step-ahead prediction intervals of the conditional expectiles at the 1\% level are more precise and tighter for every currency, yielding reasonable capital requirements. In particular, the prediction intervals of the value at risk at the 1\% level are between 11.4\$ and 50.9\% wider, whereas those of the expected shortfall at the 1\% level are between 122.1\% and 183.9\% wider. In addition, we explore the one-to-one mapping between quantiles and expectiles to assess how cryptocurrency risks evolve over time through the lens of the expected gain-loss ratio. Our empirical analyses contribute to a better understanding of cryptocurrency markets, complementing previous studies in the literature. For instance, \cn{zhang2021downside} examine how downside risk affects the cross-section of cryptocurrency returns, whereas \cnm{MS19} (2019, \cy{MS20}) investigate price discovery and arbitrage opportunities in cryptomarkets, respectively.

The remainder of this paper proceeds as follows. Section \ref{sec:ALS} discusses the main aspects of expectile-based risk measures. Section \ref{sec:model} introduces the conditional scale model and the two-step estimators of conditional and unconditional expectiles. Section \ref{sec:estimation} derives the asymptotic theory, while Section \ref{sec:mc} reports Monte Carlo simulations that assess the finite-sample behavior of the two-step estimators and their standard errors. Section \ref{sec:crypto_application_clt} examines tail risk in cryptomarkets. Section \ref{sec:conclusion} offers some concluding remarks. Appendix A collects technical Assumptions~and proofs.

\section{Expectiles}\label{sec:ALS}

Let the loss $L\in\mathbb{R}$ be a square integrable random variable with loss distribution function $F_L$. \cn{newey1987asymmetric} define the $\tau$-th expectile $\xp$ as
\begin{equation}\label{eq:expectiles_optimization_problem}
\xp=\underset{\xi\in\mathbb{R}}{\arg\min}\int_{\mathbb{R}}|\tau-\mathbf{1}(\ell<\xi)|(\ell-\xi)^2\textrm{d}F_L(\ell)=\underset{\theta\in\mathbb{R}}{\arg\min}\,\E\left[\rho_\tau(L-\xi)\right]
\end{equation}
where $\rho_\tau(u)=|\tau-\mathbf{1}(u\le 0)|u^2$ is the expectile check function and $\mathbf{1}(A)$ denotes the indicator function that takes value one if $A$ is true, zero otherwise. Apart from coinciding with the mean for $\tau=1/2$, the expectile at any $\tau\in(0,1)$ is well defined and unique for any integrable random variable \cite{newey1987asymmetric,abdous1995relating}.

Replacing the absolute deviation in the quantile check function by a quadratic deviation facilitates optimization in a substantial manner. The quantile objective function is not continuously differentiable at zero, whereby numerical implementation occasionally leads to quantile-crossing functions. In addition, quantile estimators consider only the relative frequency of observations above or below their corresponding predictions \cite{daouia2018estimation,daouia2019extremiles}, with asymptotic distributions that strongly depend on the density function and smoothing parameters \cite{cheng1997unified}. In contrast, the expectile objective function is continuously differentiable almost everywhere. Implementation is straightforward by iterative reweighted least squares \cite{daouia2018estimation}, whereas asymptotic normality requires only finite second moments \cite{holzmann2016expectile}.

The first-order condition of \eqref{eq:expectiles_optimization_problem} also indicates that expectiles depend both on the probability and magnitude of tail realizations:
\begin{displaymath}
\begin{aligned}
0&=\left.\frac{\textrm{d}}{\textrm{d}\xi}\int_{\mathbb{R}}|\mathbf{1}(\ell\le\xi)-\tau|(\ell-\xi)^2\,\textrm{d} F_L(\ell)\right|_{\xi=\xp}\\
&=\left.\frac{\textrm{d}}{\textrm{d}\xi}\int_{-\infty}^\xi(1-\tau)(\ell-\xi)^2\,\textrm{d}F_L(\ell) \right|_{\xi=\xp}+\left.\frac{\textrm{d}}{\textrm{d}\xi}\int_\xi^\infty\tau(\ell-\xi)^2\,\textrm{d}F_L(\ell)\right|_{\xi=\xp}\\
&=2(1-\tau)\int_{-\infty}^{\xp}(\xp-\ell)\,\textrm{d}F_L(\ell) +2\tau\int_{\xp}^\infty(\xp-\ell)\,\textrm{d}F_L(\ell).
\end{aligned}    
\end{displaymath}
It then follows for $u^+=\max\{u,0\}$ and $u^-=\min\{u,0\}$  that
\begin{displaymath}
\tau=\frac{\E(L-\xp)^-}{\E(L-\xp)}=\frac{\int_{-\infty} ^{\xp}(\xp-\ell) \,\textrm{d}F_L(\ell)}{\int_\mathbb{R}(\xp-\ell)\,\textrm{d}F_L(\ell)},
\end{displaymath}
implying that the expectile corresponds to the ratio of the average deviation of $Y$ below $\xp$ to the overall average deviation \cite{kuan2009assessing}. As such, evaluating \pc{keating2002universal} omega ratio at the expectile yields $\Omega_Y(\xp)=(1-\tau)/\tau$. See \cn{remillard2013statistical}, \cn{bellini2017risk} and \cn{bellini2018expectiles} for more links between expectiles and expected gain-loss ratios.

Expectiles also relate to quantiles in several aspects. First, although they both characterize the entire distribution, only the expectile function is continuous and monotonically increasing on $\tau$ for any distribution. Second, there is a straightforward link between expectiles and quantiles:
\begin{equation}\label{eq:tau equivalence}
\tau(\alpha)=\xp^{-1}(q_\alpha)=\frac{\int_{-\infty}^{q_\alpha}|\ell-q_\alpha| \,\mathrm{d}F_L(\ell)}{\int_\mathbb{R}|\ell-q_\alpha|\,\mathrm{d}F_L(\ell)},
\end{equation}
so that $\textrm{XP}_{\tau(\alpha)}=q_\alpha$ \cite{yao1996asymmetric}. Third, it is possible to estimate quantiles using expectiles, given that both are contained within the convex hull of the distribution's support \cite{waltrup2015expectile,daouia2019quantiles}. In fact, quantiles are a strict subset of the corresponding expectiles. Conveniently, this remains true in a regression context if the distribution belongs to the location-scale family \cite{yao1996asymmetric}. 

After a slow start \cite{breckling1988m,efron1991regression,jones1994expectiles,abdous1995relating,yao1996asymmetric}, the literature on expectiles has recently gained traction \cite{martin2014expectiles,bellini2017risk,kratschmer2017statistical,daouia2018estimation,daouia2020tail,girard2021extreme}. The primary reason is that the expectile is the only coherent risk measure that allows straightforward backtesting due to elicitability \cite{bellini2015elicitable,ziegel2016coherence,daouia2023expectile}.

Coherent risk measures are desirable because they satisfy the axioms of monotonicity, translation invariance, subadditivity, and positive homogeneity. Monotonicity reflects that, if one position is less risky than another, the risk measure should assign a lower risk value. The translation invariance axiom states that adding a constant amount to the position's payout should increase riskiness by exactly that amount. The subadditivity axiom ensures that there are diversification gains in pooling risks, whereas positive homogeneity dictates that multiplying the position by a positive constant should scale the risk measure by the same constant. See \cn{mcneil2015quantitative} for more details. Mathematical properties aside, we must always keep in mind that, in practice, we have to assess risk measures empirically through estimates and/or forecasts. This is exactly the idea of elicitability, which requires the feasibility of backtesting risk measures through minimization of expected scores \cite{gneiting2011,bellini2015elicitable,ziegel2016coherence}.

\subsection{Comparison between expectiles and quantile-based risk measures}

In this section, we compare expectiles with traditional quantile-based risk measures. We start with the value-at-risk measure $\operatorname{VaR}_\alpha$ at level $\alpha$, which reads
\begin{displaymath}
\operatorname{VaR}_\alpha=q_\alpha=F_L^{-1}(\alpha)=\inf\{\ell\in\mathbb{R}:\,F_L(\ell)\ge\alpha\},
\end{displaymath}
where $q_\alpha$ is the quantile function at level $\alpha\in(0,1)$. By definition, quantiles automatically satisfy the axioms of monotonicity, translation invariance, and positive homogeneity. In addition, they are elicitable for strictly increasing distribution functions \cite{THOMSON1979360,saerens2000building}. Unfortunately, value at risk does not account for the expected loss and diversification gains \cite{danielsson2001academic}.

\cn{artzner1999coherent}, \cn{acerbi2001expected}, \cn{rockafellar2002conditional} argue that the expected shortfall (ES), as defined by the expected loss given that it exceeds the value at risk, addresses both weaknesses of the VaR measure. For any integrable loss $L$ with distribution $F_L$, the ES at level $\alpha\in(0,1)$ is
\begin{displaymath}
\mathrm{ES}_\alpha=\E\left[L\mid L>\operatorname{VaR}_\alpha(L)\right].
\end{displaymath}
In particular, ES is a severity-based risk measure that belongs to the class of spectral risk measures \cite{acerbi2002spectral}. However, it fails in two accounts. First, ES is not elicitable by itself, so that backtesting is far from straightforward \cite{gneiting2011}. Second, it yields very imprecise estimates in finite samples for large values of $\alpha$ \cite{hull2014shortfalls}, especially for heavy-tailed loss distributions \cite{yamai2002comparative}.

Both value at risk and expected shortfall depend heavily on the shape of the tail, though. While the ES is too conservative because it is conditional only on tail events, the VaR is too lenient because it does not account for their severity. As such, quantile-based risk measures either underestimate or overestimate the risk exposure of a position \cite{kuan2009assessing,daouia2019extremiles}.

\cn{bellini2017risk} interpret the expectile risk measure as the amount of capital that should be added to a position to produce a sufficiently high expected gain–loss ratio. As the only risk measure that meets the conditions for coherence, law invariance and elicitability \cite{ziegel2016coherence}, expectiles are extremely convenient for modeling, forecasting and backtesting purposes. In particular, the key advantage of expectiles over VaR and ES is the fact that ALS estimation uses the available data more efficiently, by exploiting both the severity and probability of tail events \cite{daouia2018estimation}. Both VaR and ES estimates disregard the shape of the loss distribution to the left of the corresponding quantile. In addition, the precision of the ES estimates depends heavily on $\alpha$, sample size, and tail thickness. This obviously poses a problem for establishing capital requirements in the case of highly volatile and heavy-tailed asset returns \cite{yamai2002comparative,hull2014shortfalls}.

For $\tau=\alpha>1/2$, expectiles are always below both value-at-risk and expected-shortfall measures for every sample size, regardless of tail thickness. However, this ignores the expectile-quantile mapping in \eqref{eq:tau equivalence}. For instance, the first percentile of the standard Gaussian distribution is close to its expectile at $\tau=0.00145$ \cite{bellini2017risk,nolde2017elicitability}. \cn{chen2018exactitude} shows that, if $\tau(\alpha)$ is such that $\textrm{XP}_{\tau(\alpha)}=\textrm{VaR}_\alpha$, then 
\begin{equation}
\label{eq:expectile as expected shortfall}
\tau(\alpha)=\frac{\alpha\,(\textrm{ES}_\alpha-\textrm{VaR}_\alpha)}{\textrm{VaR}_\alpha+2\alpha\,(\textrm{ES}_\alpha-\textrm{VaR}_\alpha)}
\end{equation}
provided that the loss distribution has mean zero. Straightforward manipulations then yield 
\begin{equation}
\label{eq:expectile-based expected shortfall}
\textrm{ES}_{\alpha}=\textrm{VaR}_{\alpha}\left(1+\frac{1}{\alpha(\Omega_{\alpha}-1)}\right).
\end{equation}
This makes the connection with the omega ratio \cite{taylor2022forecasting}. Solving for $\Omega_\alpha$ gives way to
\begin{equation}
\label{eq:omega value-at-risk}
\Omega_\alpha=1+\frac{\textrm{VaR}_{\alpha}}{\alpha(\textrm{ES}_{\alpha}-\textrm{VaR}_{\alpha})},
\end{equation}
allowing us to compute the expected gain-loss ratio as a function of the quantile for some fixed level $\alpha$. In addition, it makes clear that, for a fixed $\alpha$, we should expect the gain-loss ratio to increase (or decrease) as the gap between expected shortfall and value at risk shrinks (or enlarges, respectively). As such, we can infer by means of the omega ratio whether a given value-at-risk model is conservative (or permissive).

In summary, expectiles offer a different perspective of the loss distribution than quantile-based risk measures. However, the literature on ALS estimation is mostly under random sampling. Only a few studies address the estimation of conditional expectiles in a time-series context \cite{taylor2008estimating,kuan2009assessing,bellini2017risk,girard2021extreme}. In the next section, we discuss a two-step estimator of the conditional expectile in a similar context to \cn{francq2015risk}. In particular, due to translation invariance and positive homogeneity, it is straightforward to compute conditional expectiles in conditional scale models by first estimating the conditional volatility of asset returns and then computing the unconditional expectiles of their standardized residuals.

\subsection{Conditional expectiles}\label{sec:model}

As expectiles are coherent risk measures, they are stable under affine transformations. This means that, in a location-scale model, the conditional expectile at a fixed level $\tau$ depends exclusively on the conditional mean and volatility, and of the $\tau$-th expectile of the innovation. As it turns out, asset returns typically exhibit a mean close to zero, but a highly persistent time-varying volatility with leverage effects \cite{bollerslev1994arch} that leads to skewness and heavy tails in asset returns \cite{fama1965behavior}. Accordingly, we entertain a GARCH-type approach to model continuously compounded returns $\{y_t\}$:
\begin{equation}\label{eq:conditional volatility model}
y_{t+1}=\sigma_t\,\eta_{t+1},\qquad\qquad\text{with}~~\sigma_t(\bs{\theta}_0)=\sigma(y_t,y_{t-1},\ldots;\,\bs{\theta}_0)~~\text{for}~~t\in\mathbb{Z},
\end{equation}
where $\{\eta_t\}$ is a sequence of iid random variables with zero mean and unit variance, independent of past returns (i.e., $y_s\perp\!\!\!\perp\eta_t$ for $s<t$), and $\bs{\theta}_0\in\mathbb{R}^m$ is a vector of model parameters. The sequence $\{y_t\}$ is a strictly stationary and ergodic solution to \eqref{eq:conditional volatility model}. This setting is very general, nesting the most popular GARCH-type models in the literature. For example, \pc{bollerslev1986generalized} GARCH(1,1) model is such that $\sigma_t^2=\omega_0+\alpha_0\,y_t^2+\beta_0\,\sigma_{t-1}^2$, where $\bs{\theta}_0=(\omega_0,\alpha_0,\beta_0)'$.

Given our interest in expectiles, we henceforth assume that log-returns $y_{t+1}$ are integrable and follow a strictly stationary process adapted to its natural filtration $\{\mathcal{F}_t\}$. The conditional expectile at level $\tau$ then reads
\begin{equation}\label{eq:conditional expectile}
\xp(y_{t+1}|\mathcal{F}_t)=\sigma_t(\bs{\theta}_0)\,\xp^\eta,
\end{equation}
where $\xp^\eta$ is the innovation expectile for $\tau\in(0,1)$.

\section{Estimation of conditional expectiles}\label{sec:estimation}

If we could observe the true innovations, their empirical expectile at the level $\tau$ would solve
\begin{displaymath}
    \widehat\xi_n^{\,0}=\underset{\xi\in\mathbb R}{\arg\min}\,\frac{1}{2n}\sum_{t=1}^n\big|\tau-\mathbf 1(\eta_t<\xi)\big|\,(\eta_t-\xi)^2.
\end{displaymath}
Equivalently, $\widehat\xi_n^{\,0}$ is the unique solution of $n^{-1}\sum_{t=1}^n\phi_\tau(\eta_t-\xi)=0$, where $\phi_\tau(u)=u\,w_\tau(u)$ with $w_\tau(u)=\big|\mathbf 1(u<0)-\tau\big|$ for $u\in\mathbb R$. This estimator is strongly consistent as long as the first moment is finite \cite{holzmann2016expectile,kratschmer2017statistical}. Unfortunately, this estimator is infeasible in our case as we do not observe the true innovations. We therefore proceed in two steps as in \cn{francq2015risk}. We first estimate $\bs\theta_0$ in the conditional variance specification by quasi-maximum likelihood and then estimate the expectile using the standardized residuals 
\begin{displaymath}
\widehat\eta_t=\frac{y_t}{\widetilde\sigma_t(\widehat{\bs\theta}_n)},\qquad t=1,\ldots,n,
\end{displaymath}
where $\widetilde\sigma_t$ considers arbitrary fixed initial values.

Let $h$ denote the instrumental density in the QML estimation, $g(y,s)=\log\big(s^{-1}h(y/s)\big)$, and
\begin{displaymath}
    \widetilde G_n(\bs\theta)=\frac{1}{n}\sum_{t=1}^n g\big(y_t,\widetilde\sigma_t(\bs\theta)\big).
\end{displaymath}
The first-step QML estimator is then $\widehat{\bs\theta}_n=\underset{\bs\theta\in\Theta}{\arg\max}\,\widetilde G_n(\bs\theta)$. In addition, for generic $(\bs\theta,\xi)$, let $\widetilde Q_n(\xi,\bs\theta)=\frac1n\sum_{t=1}^n\widetilde\psi_{t,\tau}(\bs\theta,\xi)$, with $\widetilde\psi_{t,\tau}(\bs\theta,\xi)=\phi_\tau\Big(\frac{y_t}{\widetilde\sigma_t(\bs\theta)}-\xi\Big)$. For every $\bs\theta\in\Theta$, let $\xpb{\bs\theta}$ denote the unique zero of $\widetilde Q_n(\cdot,\bs\theta)$, which we estimate by $\widehat\xi_n:=\xpb{\widehat{\bs\theta}_n}$ using the standardized residuals.

\begin{theorem}[Consistency] \label{thm:consistency_expectiles}
Let Assumptions~\ref{ass:innovation_process} to~\ref{ass:differentiability_instrumental} hold with $r=2$ in Assumption~\ref{ass:moments}. It then follows that $\widehat{\bs\theta}_n\asconv\bs\theta_0$ and $\widehat\xi_n\asconv\xi_0=\xp^\eta$.
\end{theorem}

Consistency follows under standard conditions. Assumption~\ref{ass:innovation_process} requires that innovations are independent and identically distributed (iid) with mean zero and unit variance, and that the innovation distribution is continuous at the expectile $\xi_0$. Assumptions~\ref{ass:stationarity_mixing} and~\ref{ass:compactness} dictate respectively that $y_t$ is a stationary and ergodic process, whereas $\bs\theta_0$ lies in the interior of a compact parameter space. Assumption~\ref{ass:volatility_process} constrains the volatility process to secure identification and differentiability. Assumption~\ref{ass:approximation} ensures that we can approximate the volatility process using a finite history. Assumption~\ref{ass:moments} requires moment conditions of order $r\ge2$. Finally, the identification and smoothness conditions in Assumptions~\ref{ass:identification_instrumental} and~\ref{ass:differentiability_instrumental} automatically hold for the Gaussian quasi-likelihood. See Appendix~\ref{sec:proof} for more details.

Next, we quantify how replacing the innovations by standardized residuals affects the estimation of the unconditional expectile in the second step. At the true parameter values, let $\bs{D}_t=\partial_{\bs\theta}\log\sigma_t(\bs\theta_0)$, $\bs{J}_0=\E(\bs{D}_t)$, $\psi_t=\phi_\tau(\eta_t-\xi_0)$, and $\Psi_0=\tau\{1-F_\eta(\xi_0)\}+(1-\tau)F_\eta(\xi_0)$. Define also $a(u)=\left.\frac{\partial}{\partial s}\,g(u,s)\right|_{s=1}$, $\bs{s}_t=a(\eta_t)\bs{D}_t$, $\bs{H}_0=-\E\left[\partial_{\bs\theta\bs\theta'}^2g\big(y_t,\sigma_t(\bs\theta_0)\big)\right]$, $\bs\Sigma_s=\E(\bs{s}_t\bs{s}_t')$, $\bs\sigma_{s\psi}=\E(\bs{s}_t\psi_t)$, $\bs\Sigma_\theta=\bs{H}_0^{-1}\bs\Sigma_s\bs{H}_0^{-1}$, $\bs\sigma_{\psi\theta}=\bs{H}_0^{-1}\bs\sigma_{s\psi}$, and $\sigma_\psi^2=\E(\psi_t^2)$. Asymptotic normality ensues if we strengthen the moment conditions: namely, by setting $r=4$ in Assumption~\ref{ass:moments} and by imposing Assumption~\ref{ass:qml_moments} on the moments of the QML score and curvature. 

\begin{theorem}[Asymptotic normality] \label{thm:asymptotic_expectiles}
Let Assumptions~\ref{ass:innovation_process} to~\ref{ass:qml_moments} hold with $r=4$ in Assumption~\ref{ass:moments}. It then follows that $\sqrt{n}(\widehat\xi_n-\xi_0)\dconv\mathcal N(0,V_\xi)$, where
\begin{displaymath}
 V_\xi=\frac{\sigma_\psi^2}{\Psi_0^2}+\xi_0^2\bs{J}_0'\bs\Sigma_\theta\bs{J}_0-2\frac{\xi_0}{\Psi_0}\,\bs{J}_0'\bs\sigma_{\psi\theta}.
\end{displaymath}
\end{theorem}

The first term in $V_\xi$ is the variance that would arise if we could observe the true innovations. The remaining terms capture the effects of estimating the parameters in the conditional variance. The latter involves not only the sampling error in the estimation of $\bs\theta_0$, but also how it correlates with the sampling error in the estimation of the unconditional expectile $\xi_0$.

To simplify notation, denote by $c_{n+1}=\sigma_{n+1}(\bs\theta_0)\xi_0$ the conditional expectile $\xp(y_{t+1}|\mathcal{F}_t)$ at time $n+1$ given the current state at time $n$. To conduct inference on $\widehat{c}_{n+1}=\widetilde\sigma_{n+1}(\widehat{\bs\theta}_n)\widehat\xi_n$, let $D_{n+1}=\partial_{\bs\theta}\log\sigma_{n+1}(\bs\theta_0)$ and, for the finite-memory approximation in Assumption~\ref{ass:forecast}, let $D_{n+1}^{[m]}=\partial_{\bs\theta}\log\sigma_{n+1}^{[m]}(\bs\theta_0)$. Choose a deterministic sequence $m_n$ such that $m_n\to\infty$, $\frac{m_n}{n}\to0$, and define $\mathcal H_n=\sigma(\eta_{n-m_n+1}, \ldots,\eta_n)$. Let $d_{BL}$ denote the bounded-Lipschitz distance between probability laws on $\mathbb R$.

\begin{theorem}[Gaussian prediction intervals]\label{thm:conditional_expectiles}
Let Assumptions~\ref{ass:innovation_process} to~\ref{ass:qml_moments} and~\ref{ass:forecast} hold with $r=4$ in Assumption~\ref{ass:moments}, and define
\begin{eqnarray*}
    V_{n+1}&=&\sigma_{n+1}^2(\bs\theta_0)\left[\Psi_0^{-2}\sigma_\psi^2+\xi_0^2(\bs{D}_{n+1}-\bs{J}_0)'\bs\Sigma_\theta (\bs{D}_{n+1}-\bs{J}_0)+2\xi_0\Psi_0^{-1}(\bs{D}_{n+1}-\bs{J}_0)'\bs\Sigma_{\psi\theta}\right]\\
    V_{n+1}^{[m_n]}&=&\big(\sigma_{n+1}^{[m_n]}(\bs\theta_0)\big)^2\Big[\Psi_0^{-2}\sigma_\psi^2 +\xi_0^2(\bs{D}_{n+1}^{[m_n]}-\bs{J}_0)'\bs\Sigma_\theta(\bs{D}_{n+1}^{[m_n]}-\bs{J}_0)+2\xi_0\Psi_0^{-1} (\bs{D}_{n+1}^{[m_n]}-\bs{J}_0)'\bs\Sigma_{\psi\theta}\Big].
\end{eqnarray*}
It then follows that $d_{BL}\Big(\mathcal{L}\left\{\sqrt{n}(\widehat{c}_{n+1}-c_{n+1})\,\big|\,\mathcal{H}_n\right\},\mathcal{N}(0,V_{n+1}^{[m_n]})\Big)\pconv 0$. Moreover, it also holds that $V_{n+1}^{[m_n]}-V_{n+1}\pconv 0$, so that $d_{BL}\Big(\mathcal{L}\left\{\sqrt{n}(\widehat{c}_{n+1}-c_{n+1})\,\big|\, \mathcal{H}_n\right\},\mathcal{N}(0,V_{n+1})\Big)\pconv 0$.
\end{theorem}

Theorem~\ref{thm:conditional_expectiles} strengthens the finite-history approximation in Assumption~\ref{ass:approximation} by imposing the uniform contraction in Assumption~\ref{ass:forecast} to obtain asymptotically-valid Wald confidence intervals for one-step-ahead forecasts of the conditional expectile. Although they adequately account for the first-step estimation of the scale parameters, they are infeasible. Proposition~\ref{prop:variance_estimation} in Appendix~\ref{sec:proof} provides consistent estimators for each component of the asymptotic variance, which we denote by adding a hat. We next show that the resulting Wald statistic weakly converges to a standard Gaussian distribution, enabling us to compute feasible confidence intervals for the conditional expectiles. To do so, we introduce only one additional regularity condition, which targets the plug-in estimators of the asymptotic variance components (see Assumption~\ref{ass:plugin} in Appendix~\ref{sec:proof}).

\begin{corollary}[Feasible prediction interval for conditional expectiles]\label{cor:conditional_expectile_interval}
Suppose that Assumptions~\ref{ass:innovation_process} to~\ref{ass:forecast} hold with $r=4$ in Assumption~\ref{ass:moments}, and assume that $\bs\Sigma_Y$ is positive definite. Define $\widehat {V}_{n+1}=\widetilde\sigma_{n+1}^2(\widehat{\bs\theta}_n)\Big[ \widehat\Psi_n^{-2}\widehat\sigma_\psi^2+\widehat\xi_n^2(\widehat{\bs D}_{n+1}-\widehat{\bs J}_n)'\widehat{\bs\Sigma}_\theta(\widehat{\bs D}_{n+1}-\widehat{\bs J}_n)+2\widehat\xi_n\widehat\Psi_n^{-1}(\widehat{\bs D}_{n+1}-\widehat{\bs J}_n)'\widehat{\bs\sigma}_{\psi\theta}\Big]$, where $\widehat{\bs D}_{n+1}=\partial_{\bs\theta}\log\widetilde\sigma_{n+1}(\widehat{\bs\theta}_n)$. It then follows that $\widehat {V}_{n+1}-{V}_{n+1}\pconv0$ and that
\[
 d_{BL}\left(\mathcal{L}\left\{\frac{\sqrt{n}(\widehat{c}_{n+1}-c_{n+1})}{\widehat{V}_{n+1}^{1/2}} \middle|\mathcal{\bs H}_n\right\},~\mathcal{N}(0,1)\right)\pconv0.
\]
This means that $\Pr\left\{ c_{n+1}\in\mathcal{I}_{n+1}(1-\alpha)\,\middle|\,\mathcal{H}_n\right\}\pconv 1-\alpha$, where 
\[
 \mathcal I_{n+1}(1-\alpha)=\left[\widehat{c}_{n+1}-z_{1-\alpha/2}\sqrt{\widehat{V}_{n+1}/n},~\widehat{c}_{n+1}+z_{1-\alpha/2} \sqrt{\widehat{V}_{n+1}/n}\right]
\]
with $z_{1-\alpha/2}$ denoting the $(1-\alpha/2)$-quantile of the standard Gaussian distribution.
\end{corollary}

In the next section, we evaluate the finite-sample distribution of the two-step estimator, paying close attention to whether the asymptotic variance estimator accurately quantifies the uncertainty in the expectile estimates.

\section{Monte Carlo study}\label{sec:mc}

We next assess how informative our asymptotic results are for the finite-sample behavior of the two-step estimators of both conditional and unconditional expectiles. In particular, we study the effects of sample size, volatility persistence, tail thickness, and leverage effects. The experiments distinguish inference for the unconditional expectile of the innovations from that for the more economically relevant one-step-ahead conditional expectile.

We generate returns according to
\begin{equation}
    y_t=\sigma_t(\bs\theta_0)\eta_t,
\end{equation}
where $\{\eta_t\}$ is an independent and identically distributed sequence with zero mean and unit variance. The volatility process follows a GJR-GARCH(1,1) specification:
\begin{equation}
    \sigma_t^2(\bs\theta_0)=\omega_0+\alpha_0 y_{t-1}^2+\gamma_0 y_{t-1}^2\mathbf{1}(y_{t-1}<0)+\beta_0\sigma_{t-1}^2(\bs\theta_0),
    \label{eq:mc-gjr}
\end{equation}
where we calibrate $\omega_0$ to match an unconditional volatility of 20\% per year, assuming 252 trading days; $\alpha_0=0.05$; and $p_0=\alpha_0+\beta_0+\gamma_0/2$ dictates the amount of persistence under symmetric innovations. We entertain low and high levels of persistence by setting $p_0\in\{0.90,0.98\}$. In particular, we consider both symmetric and asymmetric responses to past returns: $\gamma_0=0$ and $\beta_0\in\{0.85,0.93\}$ for the standard GARCH, whereas $\gamma_0=0.08$ and $\beta_0\in\{0.81,0.89\}$ for the asymmetric GJR-GARCH specication. The latter is convenient because it isolates the incremental effect of leverage while fixing the overall persistence in volatility.

Finally, we draw innovations from a standard Gaussian distribution or from $t_\nu$ distributions with $\nu\in\{4,8\}$ degrees of freedom. We standardize the latter by $\sqrt{(\nu-2)/\nu}$ to ensure unit variance. The $t_8$ distribution provides a heavy-tailed benchmark design, whereas the $t_4$ distribution has finite variance but no finite fourth moment. It serves as a stress test that violates the sufficient conditions for QML inference.

For each volatility model, innovation distribution, and persistence level, we consider sample sizes of $n\in\{500,1000,2500,5000\}$ time-series data after a burn-in period of 500 observations. In each of the 10,000 replications, we estimate the correctly specified volatility model by Gaussian QML, imposing not only nonnegativity, but also $\alpha+\beta+\frac{\gamma}{2}\le0.999$, with $\gamma=0$ in the symmetric model. We then calculate the standardized residuals $\widehat\eta_t=y_t/\widetilde\sigma_t(\widehat{\bs\theta}_n)$ and their sample expectile $\widehat\xi_n\equiv\xpb{\widehat{\bs\theta}_n}$ at $\tau=0.05$. In addition, we also estimate the conditional expectile $c_{n+1}=\sigma_{n+1}(\bs\theta_0)\xi_0$ at time $n+1$ by $\widehat{c}_{n+1}=\widetilde\sigma_{n+1}(\widehat{\bs\theta}_n)\widehat\xi_n$. Finally, we compute the population expectiles of the innovations from the truncated-moment equations of the Gaussian and $t$ distributions.

We assess the bias and root mean squared error (RMSE) of the conditional and unconditional expectile estimators, as well as the average length and coverage rate of their nominal 95\% confidence intervals. Table~\ref{tab:mc-xi-wald-portrait} reveals that bias and RMSE drop significantly as the sample size increases from 500 to 1,000 time-series observations, with a bias becoming negligible for samples of at least 2,500 observations. This pattern holds for both volatility specifications, as well as for every persistence level and distribution, even for the $t_4$ innovations that do not meet the regularity conditions in the asymptotic theory. In addition, the coverage rates of the 95\% Wald confidence intervals range from 0.907 to 0.951, improving monotonically with sample size. Interestingly, they are slightly better for the asymmetric GARCH specification. In contrast, increasing persistence or tail thickness slightly worsens coverage, especially for smaller sample sizes. Finally, the standard deviation of the Wald statistic $Z_{\xi,n}=\sqrt{n}(\widehat\xi_n-\xi_0)/\widehat{V}_\xi^{1/2}$ is very close to one across replications.

\begin{table}[!ht]\vspace*{1.5em}
\caption{Performance of the two-step estimator of the unconditional expectile at $\tau=0.05$\\ {\footnotesize The target is the innovation expectile $\xi_0\equiv\xp^\eta$ at $\tau=0.05$, which we estimate by $\widehat\xi_n\equiv\xpb{\widehat{\bs\theta}_n}$. For each of the 10,000 Monte Carlo replications, we report the bias and root mean squared error (RMSE) of the expectile estimator, as well as the average length and empirical coverage of the 95\% confidence interval, for each experimental design. The latter considers both standard GARCH and GJR-GARCH specifications, with different sample sizes and persistent levels (low and high). The innovations come either from a Gaussian distribution or from a $t$-student distribution (with 4 or 8 degrees of freedom).}}
\label{tab:mc-xi-wald-portrait}
\begin{adjustbox}{width=1\textwidth}
\renewcommand{\arraystretch}{1.25}
\begin{tabular}{llr*{9}{c}}
\toprule
&&&\multicolumn{4}{c}{GARCH}&&\multicolumn{4}{c}{GJR-GARCH}\\
\cline{4-7}\cline{9-12}
distribution & persistence & sample size~~ & bias & RMSE & SD($Z_{\xi,n}$) & coverage && bias & RMSE & SD($Z_{\xi,n}$) & coverage\\
\midrule
normal &  low &   500~~ & -0.0057 & 0.0800 & 1.067 & 0.936 && -0.0048 & 0.0682 & 1.034 & 0.941\\
       &      & 1,000~~ & -0.0016 & 0.0378 & 1.017 & 0.945 && -0.0015 & 0.0373 & 1.010 & 0.945\\
       &      & 2,500~~ & -0.0005 & 0.0234 & 0.995 & 0.949 && -0.0006 & 0.0234 & 0.995 & 0.949\\
       &      & 5,000~~ & -0.0002 & 0.0166 & 0.998 & 0.949 && -0.0003 & 0.0166 & 0.998 & 0.949\\
\addlinespace[1pt]
       & high &   500~~ & -0.0242 & 0.2802 & 1.089 & 0.921 && -0.0183 & 0.1211 & 1.053 & 0.927\\
       &      & 1,000~~ & -0.0092 & 0.0398 & 1.031 & 0.935 && -0.0086 & 0.0393 & 1.018 & 0.937\\
       &      & 2,500~~ & -0.0040 & 0.0240 & 1.000 & 0.944 && -0.0036 & 0.0240 & 0.999 & 0.945\\
       &      & 5,000~~ & -0.0020 & 0.0168 & 1.000 & 0.947 && -0.0018 & 0.0168 & 1.000 & 0.948\\
\addlinespace[3pt]
$t_8$  &  low &   500~~ & -0.0075 & 0.1051 & 1.076 & 0.933 && -0.0088 & 0.1371 & 1.045 & 0.937\\
       &      & 1,000~~ & -0.0032 & 0.0468 & 1.027 & 0.942 && -0.0032 & 0.0459 & 1.017 & 0.944\\
       &      & 2,500~~ & -0.0011 & 0.0292 & 1.011 & 0.943 && -0.0013 & 0.0292 & 1.011 & 0.944\\
       &      & 5,000~~ & -0.0004 & 0.0204 & 0.994 & 0.951 && -0.0005 & 0.0204 & 0.994 & 0.950\\
\addlinespace[1pt]
       & high &   500~~ & -0.0350 & 0.4175 & 1.092 & 0.917 && -0.0230 & 0.1471 & 1.064 & 0.922\\
       &      & 1,000~~ & -0.0111 & 0.0489 & 1.035 & 0.934 && -0.0105 & 0.0484 & 1.026 & 0.939\\
       &      & 2,500~~ & -0.0046 & 0.0298 & 1.016 & 0.940 && -0.0043 & 0.0298 & 1.015 & 0.940\\
       &      & 5,000~~ & -0.0022 & 0.0207 & 0.997 & 0.950 && -0.0020 & 0.0207 & 0.997 & 0.949\\
\addlinespace[3pt]
$t_4$  &  low &   500~~ & -0.0432 & 1.1820 & 1.125 & 0.918 && -0.0317 & 1.0495 & 1.064 & 0.924\\
       &      & 1,000~~ & -0.0117 & 0.0671 & 1.065 & 0.929 && -0.0135 & 0.0653 & 1.051 & 0.930\\
       &      & 2,500~~ & -0.0064 & 0.0445 & 1.045 & 0.928 && -0.0076 & 0.0438 & 1.043 & 0.927\\
       &      & 5,000~~ & -0.0034 & 0.0325 & 1.037 & 0.935 && -0.0042 & 0.0319 & 1.035 & 0.936\\
\addlinespace[1pt]
       & high &   500~~ & -0.0820 & 2.0999 & 1.118 & 0.907 && -0.0512 & 1.4201 & 1.082 & 0.908\\
       &      & 1,000~~ & -0.0204 & 0.0701 & 1.076 & 0.917 && -0.0218 & 0.0690 & 1.061 & 0.918\\
       &      & 2,500~~ & -0.0101 & 0.0450 & 1.047 & 0.922 && -0.0111 & 0.0444 & 1.042 & 0.922\\
       &      & 5,000~~ & -0.0053 & 0.0328 & 1.041 & 0.928 && -0.0060 & 0.0320 & 1.035 & 0.929\\
\bottomrule\\
\end{tabular}\end{adjustbox}\end{table}

Table~\ref{tab:mc-ce-wald-portrait} documents the performance of the conditional expectile estimator (bias and RMSE) across Monte Carlo replications for both GARCH-type specifications. It also reports the average width and empirical coverage of the 95\% Wald confidence interval. As before, bias and RMSE shrink substantially as the sample size increases from 500 to 1,000 time-series observations. This does not depend on the experimental design, remaining true for every persistence level, distribution, and volatility specification. A similar pattern arises for the 95\% Wald confidence intervals, which show remarkably little variation in coverage (between 0.927 and 0.952). Lastly, the standard deviation of $Z_{c,n+1}$ is slightly closer to unit than the standard deviation of the Wald statistic for the unconditional quantile.

\begin{table}[!ht]\vspace*{2em}
\caption{Performance of the two-step estimator of the conditional expectile at $\tau=0.05$\\ {\footnotesize The target is the conditional expectile $c_{n+1}\equiv\sigma_{n+1}\xi_0$, which we estimate by $\widehat{c}_{n+1}\equiv\widetilde\sigma_{n+1}(\widehat{\bs\theta}_n)\widehat{\xi}_n$. For each of the 10,000 Monte Carlo replications, we report the bias and root mean squared error (RMSE) of the conditional expectile estimator, as well as the average length and empirical coverage of the 95\% prediction interval, for each experimental design. The latter considers both standard GARCH and GJR-GARCH specifications, with different sample sizes and persistent levels (low and high). The innovations come either from a Gaussian distribution or from a $t$-student distribution (with 4 or 8 degrees of freedom).}}
\label{tab:mc-ce-wald-portrait}
\begin{adjustbox}{width=1\textwidth}
\renewcommand{\arraystretch}{1.25}
\begin{tabular}{llr*{9}{c}}
\toprule
&&&\multicolumn{4}{c}{GARCH}&&\multicolumn{4}{c}{GJR-GARCH}\\
\cline{4-7}\cline{9-12}
distribution & persistence & sample size & bias & RMSE & SD($Z_{c,n+1}$) & coverage && bias & RMSE & SD($Z_{c,n+1}$) & coverage\\
\midrule
normal &  low &   500 & -0.0042 & 0.1161 & 1.046 & 0.938 && -0.0037 & 0.1273 & 1.045 & 0.940\\
       &      & 1,000 &  0.0003 & 0.0748 & 1.043 & 0.942 && -0.0001 & 0.0848 & 1.048 & 0.941\\
       &      & 2,500 & -0.0008 & 0.0477 & 1.018 & 0.951 && -0.0009 & 0.0537 & 1.008 & 0.949\\
       &      & 5,000 & -0.0002 & 0.0329 & 0.999 & 0.952 && -0.0004 & 0.0373 & 1.000 & 0.948\\
       & high &   500 & -0.0248 & 0.2199 & 1.103 & 0.928 && -0.0234 & 0.1557 & 1.064 & 0.938\\
       &      & 1,000 & -0.0105 & 0.0779 & 1.043 & 0.936 && -0.0114 & 0.0881 & 1.032 & 0.941\\
       &      & 2,500 & -0.0052 & 0.0478 & 1.007 & 0.949 && -0.0055 & 0.0551 & 1.013 & 0.947\\
       &      & 5,000 & -0.0026 & 0.0335 & 1.002 & 0.950 && -0.0026 & 0.0374 & 1.000 & 0.951\\
$t_8$  &  low &   500 & -0.0056 & 0.1560 & 1.071 & 0.935 && -0.0075 & 0.1929 & 1.069 & 0.938\\
       &      & 1,000 & -0.0024 & 0.1008 & 1.059 & 0.938 && -0.0019 & 0.1174 & 1.066 & 0.938\\
       &      & 2,500 &  0.0003 & 0.0625 & 1.031 & 0.940 && -0.0002 & 0.0722 & 1.028 & 0.943\\
       &      & 5,000 &  0.0002 & 0.0440 & 1.014 & 0.947 && -0.0001 & 0.0497 & 1.000 & 0.949\\
       & high &   500 & -0.0352 & 0.3557 & 1.137 & 0.928 && -0.0284 & 0.1904 & 1.077 & 0.939\\
       &      & 1,000 & -0.0119 & 0.1066 & 1.058 & 0.938 && -0.0142 & 0.1292 & 1.049 & 0.942\\
       &      & 2,500 & -0.0048 & 0.0625 & 1.023 & 0.946 && -0.0053 & 0.0731 & 1.018 & 0.947\\
       &      & 5,000 & -0.0022 & 0.0448 & 1.016 & 0.946 && -0.0025 & 0.0494 & 1.002 & 0.949\\
$t_4$  &  low &   500 & -0.0150 & 0.2437 & 1.107 & 0.927 && -0.0153 & 0.2566 & 1.102 & 0.931\\
       &      & 1,000 & -0.0058 & 0.1620 & 1.095 & 0.932 && -0.0075 & 0.1895 & 1.094 & 0.936\\
       &      & 2,500 & -0.0044 & 0.1104 & 1.051 & 0.942 && -0.0059 & 0.1283 & 1.038 & 0.945\\
       &      & 5,000 & -0.0015 & 0.0831 & 1.018 & 0.945 && -0.0033 & 0.1024 & 1.019 & 0.944\\
       & high &   500 & -0.0440 & 0.6450 & 1.147 & 0.927 && -0.0293 & 0.2531 & 1.097 & 0.933\\
       &      & 1,000 & -0.0168 & 0.1764 & 1.094 & 0.932 && -0.0196 & 0.2046 & 1.074 & 0.938\\
       &      & 2,500 & -0.0065 & 0.1196 & 1.020 & 0.949 && -0.0073 & 0.1258 & 1.012 & 0.949\\
       &      & 5,000 & -0.0039 & 0.1045 & 1.013 & 0.946 && -0.0057 & 0.1309 & 1.008 & 0.948\\
\bottomrule
\end{tabular}\end{adjustbox}\end{table}

This improvement relative to the unconditional expectile estimator is consistent with the large-sample theory for two-step estimators of conditional risk measures. The asymptotic variance accounts for the uncertainty both in the estimation of the innovation expectile and in the one-step-ahead volatility forecast. The influence function of the conditional expectile attenuates part of the first-step estimation error that remains visible in the unconditional expectile estimator, bringing the dispersion of the Wald statistic closer to one and coverage closer to nominal levels.

We complement the analysis in two ways. First, we compute the skewness and kurtosis of the Wald statistics $Z_{\xi,n}=\sqrt{n}(\widehat\xi_n-\xi_0)/\widehat{V}_\xi^{1/2}$ and $Z_{c,n+1}=\sqrt{n}(\widehat{c}_{n+1}-c_{n+1})/\widehat{V}_{n+1}^{1/2}$ across replications. We find that their skewness and excess kurtosis are close to zero, especially for samples of at least 1,000 observations, confirming that the Gaussian approximation performs very well in finite samples. Second, we also look at the widths of the 95\% prediction intervals of the conditional expectiles at $\tau=0.05$ relative to those of the value at risk and expected shortfall at $\alpha=0.05$, which we compute using \pc{gao2008estimation} standard errors. Figure~\ref{fig:mc_interval_widths_level} reveals that they are relatively tighter for the conditional expectile at $\tau=0.05$ than for the value at risk and, especially, expected shortfall at $\alpha=0.01$. This is across the board,  regardless of the sample size, innovation distribution, persistence level, and volatility specification. Moreover, further simulations show that fixing both quantile and expectile levels to 1\% increases the expectile advantage over the expected shortfall, especially for $t$ innovations.

Although imposing $\alpha=\tau$ is appropriate for statistical comparisons, it implies completely different capital requirements. To make them comparable, we can fix the value at risk and then compute $\tau(\alpha)$ such that the expectile and value at risk coincide, as in~\ref{eq:tau equivalence}. The regulatory guidelines of the Bank for International Settlements suggest that the expected shortfall at the 2.5\% quantile level imposes a similar capital requirement to the value at risk at 1\% (\cnm{BIS2013}, 2013, \cy{BIS2014}). We thus compare their 95\% prediction intervals to those of the conditional expectile of $\tau(0.01)$, which typically takes a more extreme value between 0.33 and 0.66 across Monte Carlo experiments. Figure~\ref{fig:mc_interval_widths_capital} shows that the average lengths of the prediction intervals are not so different for sample sizes of 500 observations, despite the large differences in expectile and quantile levels (e.g., almost tenfold between XP and ES). For larger samples, the average length of the XP prediction intervals become between 20\% and 25\% larger than those of the value at risk and expected shortfall.

\begin{figure}[!ht]
\centering\caption{Relative lengths of the prediction intervals of the tail risk measures with $\alpha=\tau=0.05$\\ 
{\footnotesize The plots summarize the replication-level ratio between the lengths of the nominal 95\% prediction intervals of the conditional expectile, value at risk and expected shortfall with $\alpha=\tau=0.05$. Top row refers to the ratios with respect to the value-at-risk interval length, whereas bottom rows to the expected shortfall.}}
\includegraphics[height=.375\textheight,trim= 0 .25cm 0 .5cm,clip]{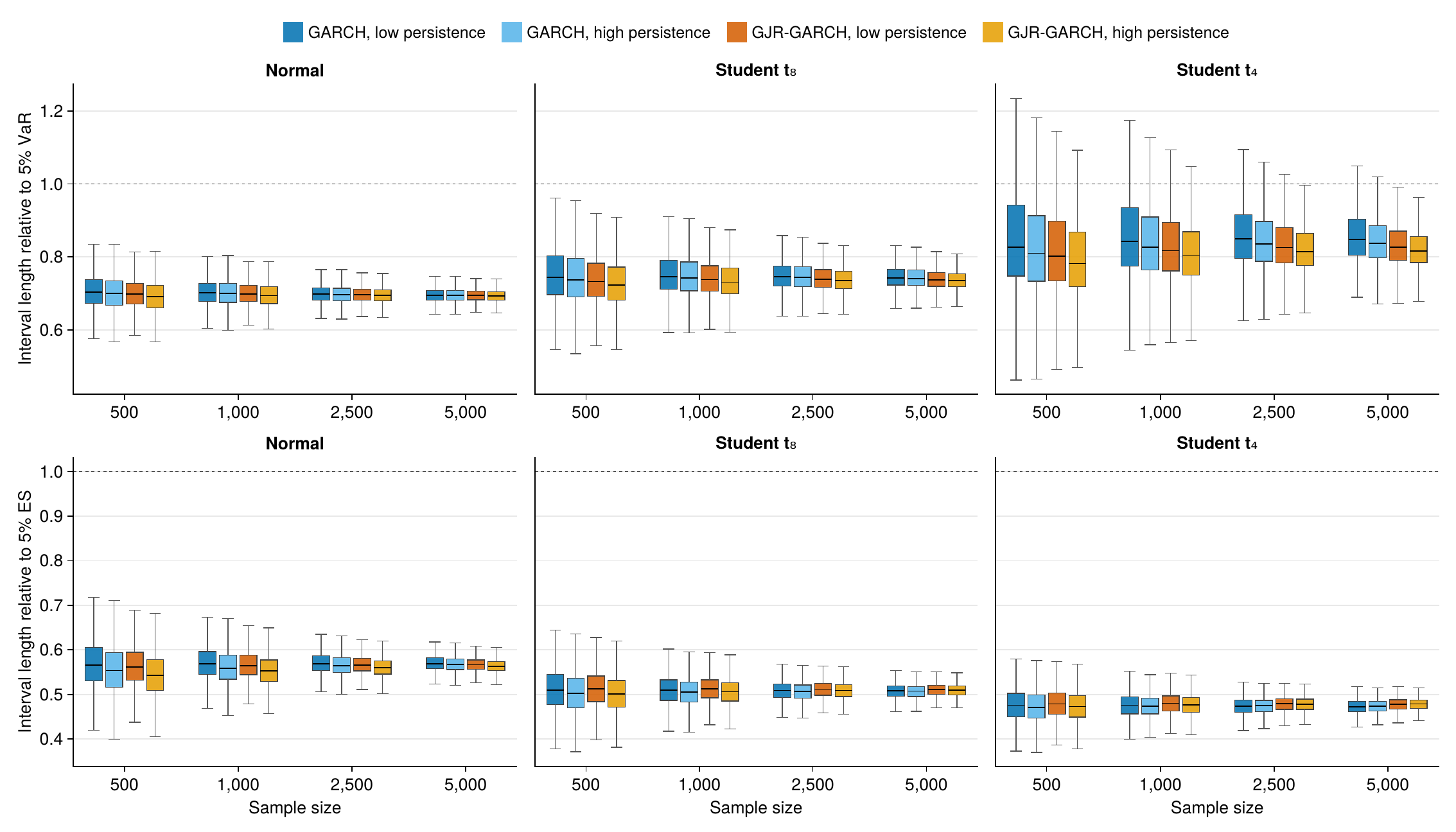}
\label{fig:mc_interval_widths_level}
\end{figure}

\begin{figure}[!ht]\vspace*{.375cm}
\centering\caption{Relative lengths of the prediction intervals given comparable capital requirements\\ 
{\footnotesize The plots summarize the replication-level ratio between the lengths of the nominal 95\% prediction intervals of the value at risk at $\alpha=0.01$, expected shortfall at $\alpha=0.025$, and conditional expectile at $\tau(0.01)$. Top row refers to the ratios with respect to the value-at-risk interval length, whereas bottom rows to the expected shortfall.}}
\includegraphics[height=.375\textheight,trim = 1.7cm 3.25cm 1.8cm 3cm,clip]{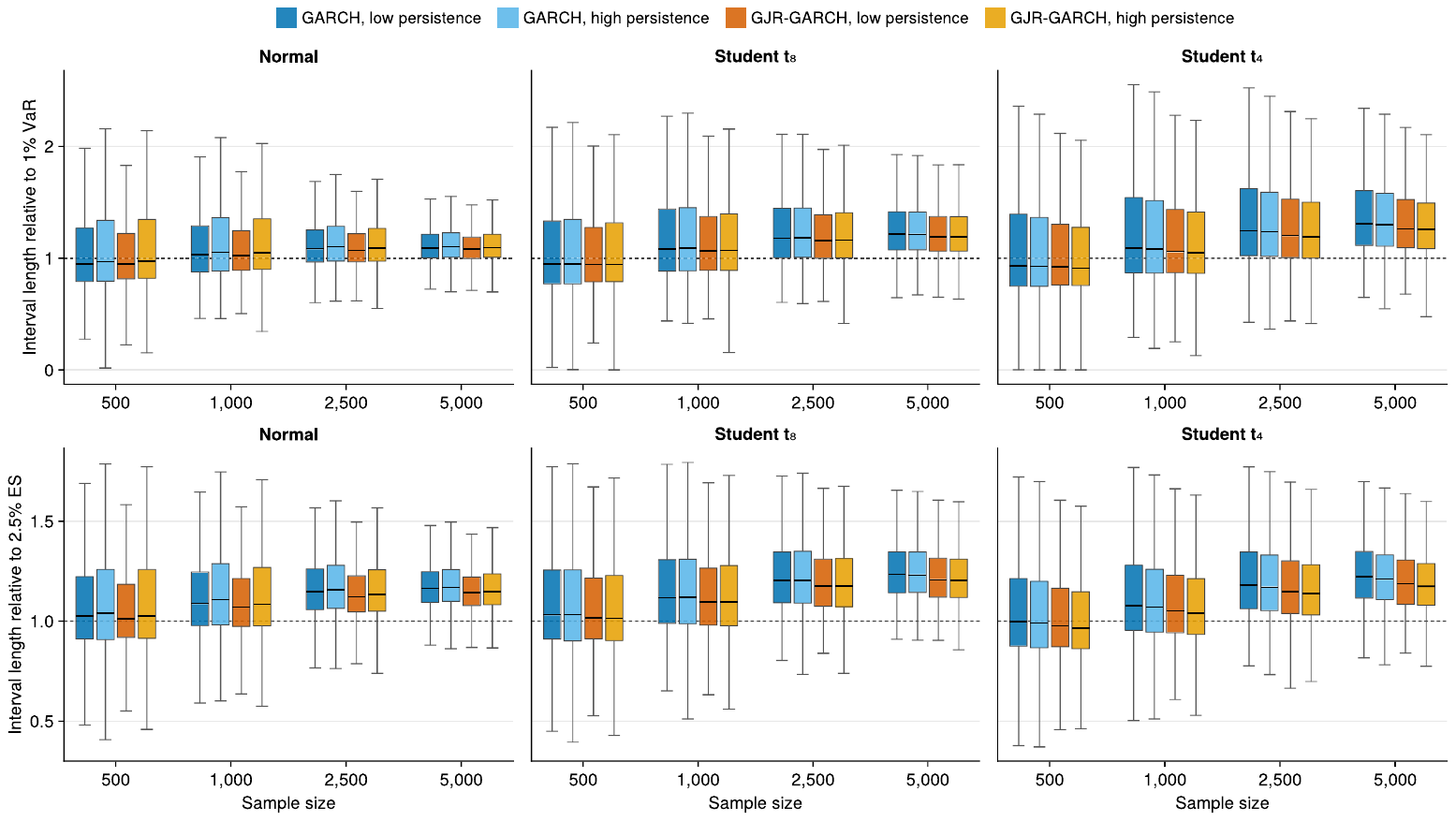}
\label{fig:mc_interval_widths_capital}
\end{figure}

\section{Tail risk in cryptomarkets}\label{sec:crypto_application_clt}

We collect data from Yahoo Finance on the daily values of the cryptocurrencies with the highest market values at the end of 2023: namely, Bitcoin (BTC), Ethereum (ETH), and Binance Coin (BNB).\footnote{~~Yahoo Finance sources their crypto data from CoinMarketCap. Quantitative results remain virtually the same using Coin Codex data, alleviating concerns with crypto data quality \cite{SSY26}.} For the sake of comparison, we also gather the daily exchange rate between the Euro and US dollar (EUR). The sample ranges from January 2016 to December 2023 for EUR and BTC, but only from November 2017 to December 2023 for BNB and ETH.

Figure~\ref{fig:drawdown} documents this difference in a more direct manner by plotting the distribution of their historical drawdowns. We calculate the drawdown at time $t$ as the difference between the maximum cumulative return up to time $t$ and the cumulative return at time $t$, so that it measures the percentage drop from the all-time high. Crypto drawdowns are much larger (and more persistent) than EUR drawdowns, which remain comparatively modest throughout the sample period.

To deal with nonstationarity in the exchange rates, we compute their daily percentage changes (in logs): i.e., $y_t=100\ln(P_t/P_{t-1})$, where $P_t$ is the value of the currency in US dollars at time $t$. Table~\ref{tab:summary} documents that sample sizes are larger for crypto assets because they trade every day of the week, even on bank holidays. However, the most striking feature is the range in which daily cryptocurrency returns dwell. Their minimum values reflect drops in value of about 50\%, with highest returns varying from 22.5\% to more than 53\%, in just one day! This variation leads to very high levels of daily volatility, from 3.75\% to 5.46\%.

In line with the risk-return tradeoff, typical daily returns are also very high, with average and median values that vary between 9\% and 23\% and between 8\% and 15\%, respectively. In comparison, log-changes in the EUR range from -2.81\% to 1.82\%, with zero mean and a daily volatility of 0.47\%. The excess kurtosis in every exchange rate is consistent with heavy tails and conditional heteroskedasticity. In turn, skewness is particularly strong for cryptocurrencies, which might indicate that volatility responds asymmetrically to positive and negative returns. These features obviously invalidate any risk assessment based exclusively on volatility, instead calling for tail risk measures.\newpage

\begin{figure}[!ht]
\centering\caption{Daily exchange rates vis-à-vis the US dollar from January 2016 to December 2023\\ 
{\footnotesize We plot the value of Binance Coin (BNB), Bitcoin (BTC), Ethereum (ETH), and Euro (EUR) in US dollar.}}
\includegraphics[width=\linewidth, trim = 1cm 4cm 1cm 3cm, clip]{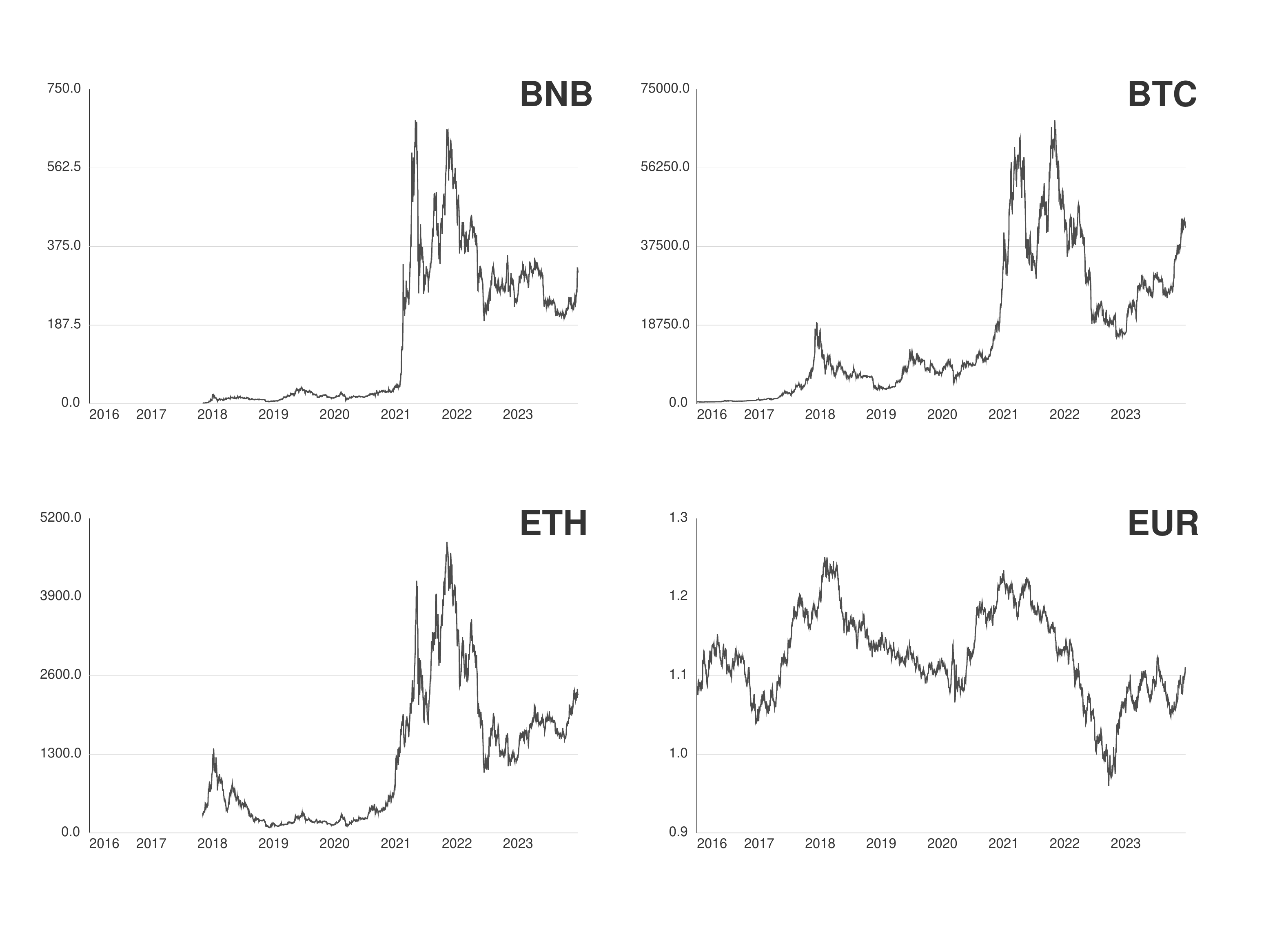}
\label{fig:prices}
\end{figure}

\begin{figure}[!ht]
\centering\caption{Box plots of the historical drawdowns from 2018 and 2023\\
{\footnotesize We plot the distributions of the historical drawdowns of the BNB, BTC, ETH and EUR exchange rates. We start as from January 2018 to have at least one year of observations for each currency in the initial sample.}}
\includegraphics[height=8.75cm,trim = 1cm 2.75cm 1cm 3cm,clip]{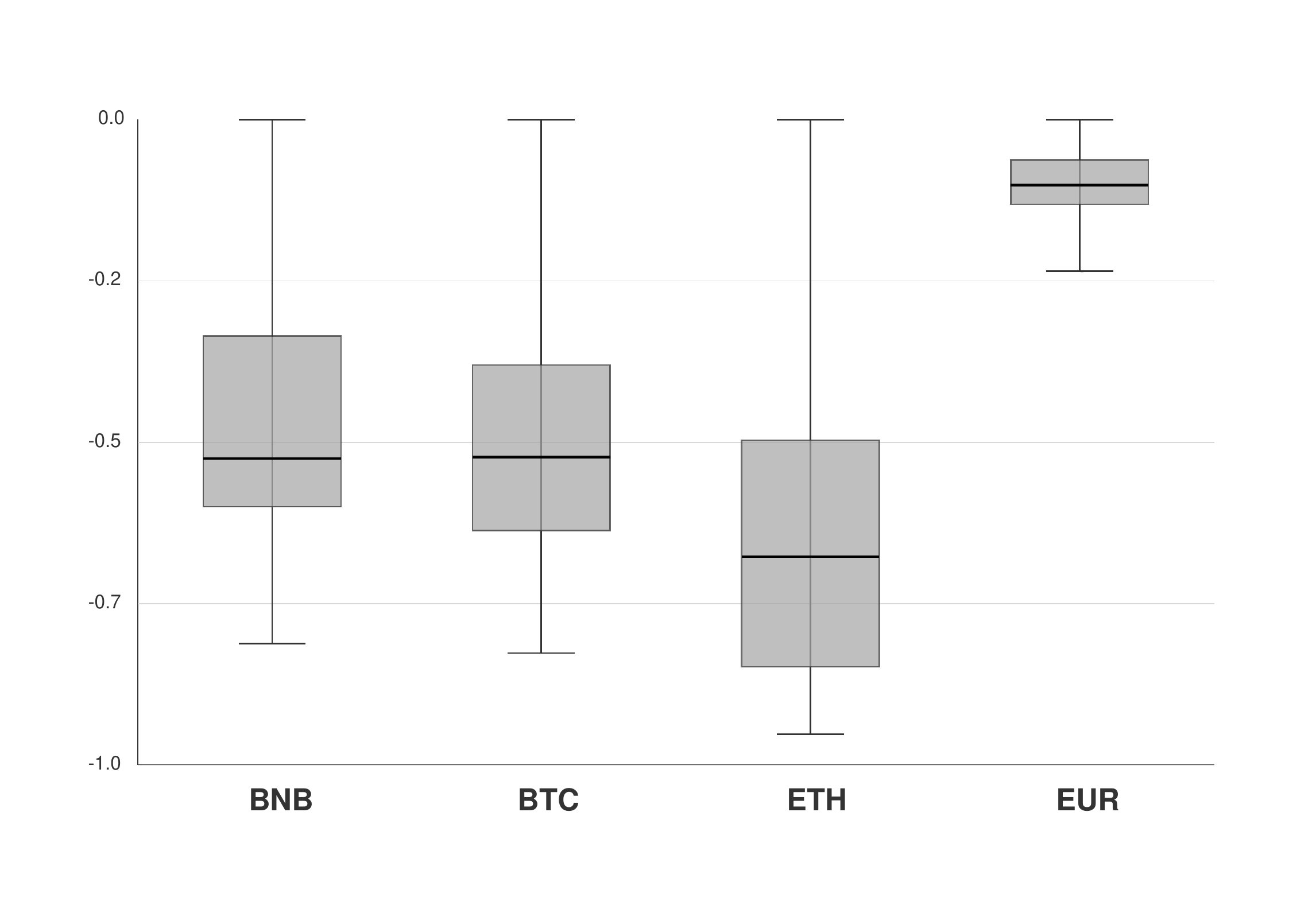}
\label{fig:drawdown}
\end{figure}

\begin{table}[!ht]
\vspace*{1em}\caption{Descriptive statistics for the daily log-changes in exchange rates\\ {\footnotesize We report the number of observations (\#obs) for each time series, as well as their main summary statistics. In particular, we display not only their minimum (min), average (mean) and maximum (max) values, but also their standard deviation (std dev), skewness (skew), kurtosis (kurt), and quartiles ($q_{0.25}$, median and $q_{0.75}$). The sample ranges from January 2016 to December 2023 for EUR and BTC, and from November 2017 to December 2023 for BNB and ETH.}}
\begin{adjustbox}{width=1\textwidth}
\begin{tabular}{p{1.85cm}crrrrrrrrr}
\hline
currency & \#obs &    min & $q_{0.25}$ & median & $q_{0.75}$ &   max & mean & std dev &   skew & kurt\\
\hline
BNB      & 2,243 & -54.31 &      -1.88 &   0.10 &       2.35 & 52.92 & 0.23 &    5.46 &   0.40 & 20.03\\
BTC      & 2,921 & -46.47 &      -1.27 &   0.15 &       1.69 & 22.51 & 0.16 &    3.75 &  -0.72 & 14.82\\
ETH      & 2,243 & -55.07 &      -1.90 &   0.08 &       2.38 & 23.47 & 0.09 &    4.80 &  -0.93 & 13.84\\
EUR      & 2,082 &  -2.81 &      -0.28 &   0.00 &       0.28 &  1.82 & 0.00 &    0.47 &  -0.05 & 4.90\\
\hline
\end{tabular}\end{adjustbox}
\label{tab:summary}
\end{table}

Further analysis shows little autocorrelation, but strong persistence in conditional volatility. Given the skewness in the crypto assets, we estimate GJR-GARCH(1,1) models for each currency using Gaussian QML over rolling windows of 1,000 observations. The asymmetric volatility specification matters because ignoring leverage effects could well distort tail forecasts and their prediction intervals. For each window, we first compute standardized residuals and then estimate not only their unconditional expectile, but also the corresponding conditional expectiles by plugging in one-step-ahead volatility forecasts. Finally, it is straightforward to compute their 90\% prediction intervals using the plug-in variance estimators.

Figure~\ref{fig:empirical confidence interval} depicts the one-step-ahead forecasts of the value at risk, expected shortfall and expectile at level $\alpha=\tau=0.01$ from January 2, 2020 to December 30, 2023. It is apparent that the expected shortfall spikes more sharply than the other tail risk measures, whereas the conditional expectile has the smoothest variation. Moreover, the latter prediction intervals are much tighter than those of the quantile-based risk measures: namely, between 66.3\% and 89.8\% of the length for the value at risk and from 35.2\% to 45.0\% of that for the expected shortfall. We assess the one-step ahead forecasts of the value at risk and expectile at $\alpha=\tau=0.01$ using statistical tests and backtest procedures. In particular, we check the value-at-risk performance using coverage and duration tests \cite{kupiec1995techniques,christoffersen1998evaluating,christoffersen2004backtesting}, whereas we conduct \pc{mcneil2000estimation} bootstrap-based residual tests to evaluate conditional expectiles.

\begin{figure}[!ht]
\caption{Conditional tail risk measures at level $\alpha=\tau=0.01$\\ 
{\footnotesize The plots in the first column display the daily percentage changes in the value of Binance Coin (BNB), Bitcoin (BTC), Ethereum (ETH), and Euro (EUR) in US dollars from January 2020 to December 2023, as well as the one-day-ahead forecasts of the value at risk, expected shortfall and expectile at level $\alpha=\tau=0.01$. In the second column, we depict how the expectile level $\tau(\alpha)$ should vary over time to make expectile and value at risk coincide.}}
\includegraphics[width=\textwidth,trim = 1cm 0 0 0,clip]{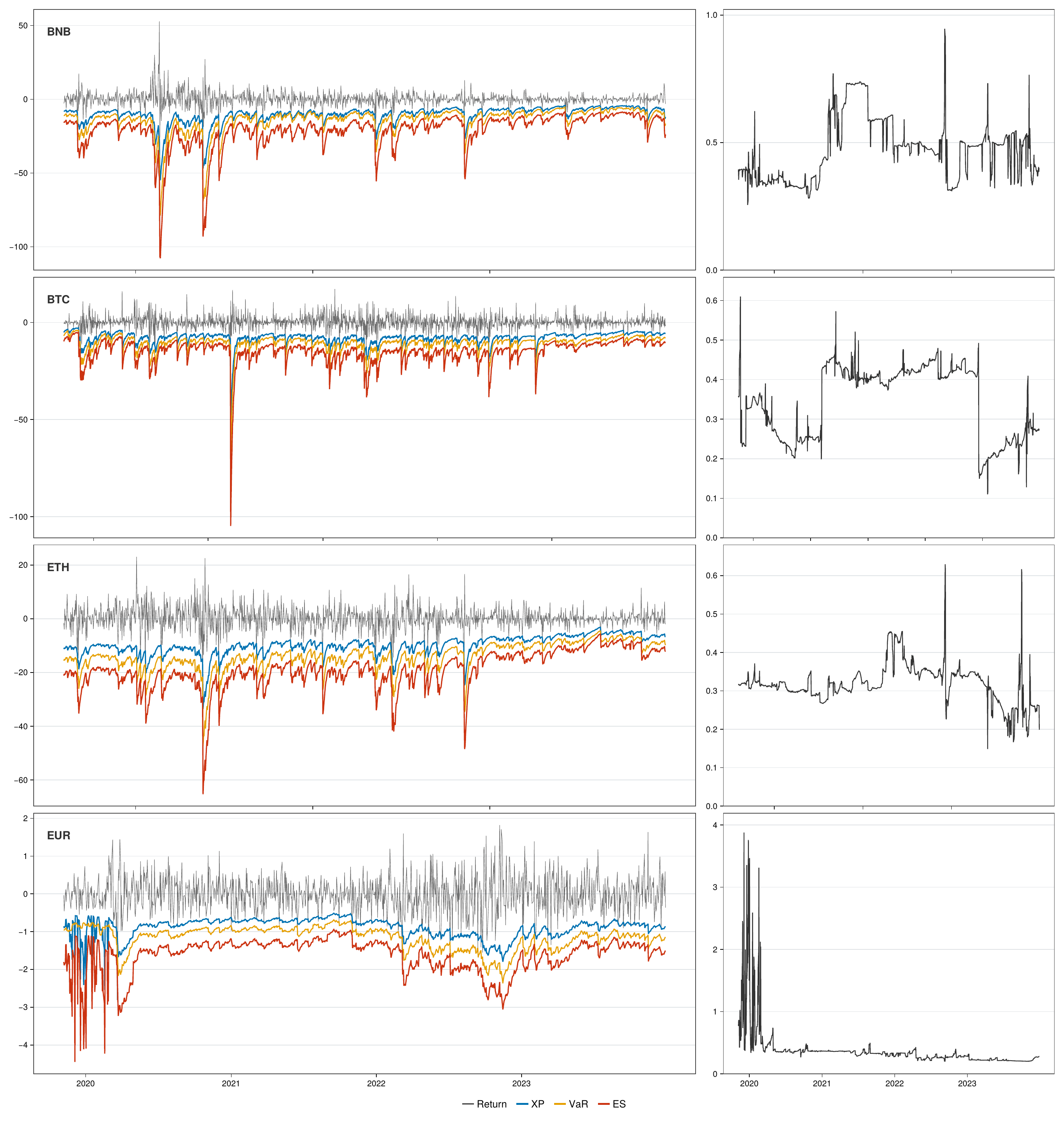}\vspace*{-1.5em}
\label{fig:empirical confidence interval}
\end{figure}\clearpage

\begin{table}[!th]\renewcommand{\arraystretch}{1.25}
\centering\caption{Frequency of value at risk and expectile forecast exceedances at $\alpha=\tau=0.01$\\ {\footnotesize We report the expected and realized number of exceedances for the one-step-ahead forecasts of the value at risk and expectile based on rolling windows of 1,000 observations. Exceedances occur if realized loss exceeds the tail risk measure. Apart from the corresponding relative frequency (i.e., number of realized exceedances over number of out-of-sample forecasts), we also document the p-values of the coverage tests for the value at risk and the bootstrap-based p-value of the residual test for the expectile.}}
\label{tab:app_backtest}
\begin{tabular}{lccrccrccrcc}
\toprule
 & \multicolumn{2}{c}{BNB} && \multicolumn{2}{c}{BTC} && \multicolumn{2}{c}{ETH} && \multicolumn{2}{c}{EUR}\\
\cline{2-3}\cline{5-6}\cline{8-9}\cline{11-12}
daily exceedances  &    VaR &     XP &&    VaR &     XP &&    VaR &     XP &&    VaR & XP\\
\midrule
expected           & 12.43  & 12.43  && 19.21  & 19.21  && 12.43  & 12.43  && 10.82  & 10.82\\
realized           & 10     & 16     && 20     & 36     && 8      & 19     && 14     & 36\\
relative frequency & 0.80\% & 1.29\% && 1.04\% & 1.87\% && 0.64\% & 1.53\% && 1.29\% & 3.33\%\\
coverage tests\\
~~~~unconditional  &   0.47 &        &&   0.86 &        &&   0.18 &        &&   0.35 &\\
~~~~conditional    &   0.71 &        &&   0.80 &        &&   0.38 &        &&   0.54 &\\
~~~~duration       &   0.03 &        &&   0.30 &        &&   0.02 &        &&   0.81 &\\
residual test      &        &   0.48 &&        &   0.79 &&        &   0.06 &&        & 0.98\\
\bottomrule
\end{tabular}\end{table}

Table~\ref{tab:app_backtest} shows that number of VaR exceedances is close to the nominal one-percent target: 1.05\% for BNB, 1.04\% for BTC, 0.72\% for ETH, and 1.20\% for EUR. As expected, there are more exceedances for the one-step-ahead forecasts of the expectiles at level $\tau=1\%$ in view that setting $\tau=\alpha$ does not impose equivalent capital requirements. Finally, the coverage and residual tests cannot reject the congruency of the one-step-ahead forecasts of the tail risk measures at the 5\% significance level.

\begin{table}[!htb]\renewcommand{\arraystretch}{1.25}
\vspace*{1em}\centering\caption{Prediction intervals of the tail risk measures with comparable capital requirements\\ {\footnotesize We consider the value at risk at the 1\% quantile level, the expected shortfall at the 2.5\% quantile level, and the expectile at the level $\tau(\alpha)=\tau(0.01)$ such that it coincides with the value at risk at level $\alpha=0.01$. We report the number of rolling windows (\#windows) we deploy to produce out-of-sample forecasts of the tail risk measures, as well as the average value of $\tau(\alpha)$. We also display not only the average difference between capital requirements and realized losses (capital buffer), but also the average width of the 95\% prediction intervals of the one-step-ahead forecasts of each tail risk measure (interval width).}}
\label{tab:comparison_risk_measures}
\begin{tabular}{lccrcccrccc}
\toprule
         &&&&     \multicolumn{3}{c}{capital buffer} && \multicolumn{3}{c}{interval width}\\
\cline{5-7}\cline{9-11}
currency~~~~ & \#windows & $\bar\tau(\alpha)=\bar\tau(0.01)$ &&   VaR &    ES &    XP &&  VaR &   ES & XP\\
\midrule
BNB          &     1,243 &                            0.47\% && 13.97 & 14.70 & 13.97 && 7.86 & 6.33 & 7.21\\
BTC          &     1,921 &                            0.34\% && 10.69 & 11.24 & 10.69 && 4.50 & 5.79 & 7.64\\
ETH          &     1,243 &                            0.32\% && 14.02 & 14.04 & 14.02 && 6.86 & 6.88 & 8.38\\
EUR          &     1,082 &                            0.38\% && 1.17  &  1.26 & 1.17  && 0.35 & 0.41 & 0.54\\
\bottomrule
\end{tabular}
\end{table}

Although useful for calibration, the above exceedance analysis does not consider comparable capital requirements given that we use the same quantile/expectile levels for each tail risk measure \cite{bellini2017risk}. Figure~\ref{fig:empirical confidence interval} indeed shows how the expectile levels of each currency must change over time to match the VaR capital requirements. We then fix the same capital requirement for the value at risk and expectile by setting $\alpha=0.01$ and $\tau(\alpha)=\tau(0.01)$ in \eqref{eq:tau equivalence}. In addition, we also compute the expected shortfall at the 2.5\% quantile level following the regulatory guidelines in \cnm{BIS2013} (2013, \cy{BIS2014}). Table~\ref{tab:comparison_risk_measures} shows not only that the average values of $\tau(\alpha)=\tau(0.01)$ vary from 0.32\% to 0.47\% across exchange rates, but also that the average differences between the ES$_{0.025}$ and VaR$_{0.01}$ capital requirements are quite small in percentage losses: 73 bps for BNB, 55 bps for BTC, 2 bps for ETH, and 9 bps for EUR. Lastly, as expected from the Monte Carlo results, the prediction intervals of the conditional expectiles are relatively wider on average than those of the value at risk and expected shortfall because $\tau(\alpha)$ is much smaller than $\alpha=0.01$.

\section{Conclusion}\label{sec:conclusion}

In this paper, we consider a two-step approach for the estimation of conditional expectiles. In particular, we first estimate the parameters of the conditional variance by quasi-maximum likelihood and then employ standardized residuals to compute the unconditional expectiles of the innovations. We quantify the impact of first-step estimation error on the asymptotic distribution of the expectile estimator, showing how to compute consistent standard errors for the unconditional expectile and prediction intervals for the conditional expectiles. Empirically, we assess how the conditional expectile performs in comparison with traditional quantile-based risk measures in the analysis of daily cryptocurrency returns.

There are at least two directions worth exploring in future research. The first is to extend our asymptotic theory to deal with nonzero expected returns \cite{francq2018estimation}. The second is to study the high-dimensional case of large portfolios \cite{francq2020virtual}. As it happens, expectiles are particularly suitable for handling several assets because they are coherent measures even under nonelliptical distributions \cite{artzner1999coherent}.

\lineskip=.5ex \baselineskip 3.75ex
\bibliography{ALS}

\newpage\appendix
\section{Technical Appendix}\label{sec:proof}
\setcounter{table}{0}
\renewcommand{\thetable}{A\arabic{table}}
\setcounter{equation}{0}
\renewcommand{\theequation}{A.\arabic{equation}}
In this appendix, we establish the large-sample properties of the two-step estimator of the conditional and unconditional expectile at a fixed level $\tau\in(0,1)$. The first step in the procedure estimates the parameters of the conditional scale model, whereas the second step estimates the unconditional expectile of the innovations using the standardized residuals. To obtain estimates for the conditional expectile, it suffices to take scale the volatility forecast by the unconditional expectile estimate.

\subsection{Conditions and notation}

Let $\mathcal{F}_t=\sigma(y_s:s\le t)$ and consider 
\begin{equation}\label{eq:model_appendix}
 y_t=\sigma_t(\bs\theta_0)\eta_t, \qquad \sigma_t(\bs\theta)=\sigma(y_{t-1},y_{t-2},\ldots;\bs\theta),  \qquad t\in\mathbb Z,
\end{equation}
where $\bs\theta_0\in\Theta\subset\mathbb{R}^m$. Let $\bs{D}_t(\bs\theta)=\partial_{\bs\theta}\log\sigma_t(\bs\theta)$, $\bs{D}_t=\bs{D}_t(\bs\theta_0)$, and $\bs{J}_0=\E[\bs{D}_t]$. For an arbitrary initialization of the pre-sample observations, denote by $\widetilde\sigma_t(\bs\theta)$ and $\widetilde{\bs{D}}_t(\bs\theta)$ the corresponding finite-history quantities. Define the stationary and finite-history standardized residuals respectively as $\eta_t(\bs\theta)=y_t/\sigma_t(\bs\theta)$ and $\widetilde\eta_t(\bs\theta)=y_t/\widetilde\sigma_t(\bs\theta)$. For $u\in\mathbb{R}$, let $\phi_\tau(u)=u w_\tau(u)$ and
\[
 w_\tau(u)=\big|\bs{1}(u<0)-\tau\big|=(1-\tau)\bs{1}(u<0)+\tau\bs{1}(u\ge0),
\]
so that the innovation expectile $\xi_0:=\xp^\eta$ is the unique solution of $\E\big[\phi_\tau(\eta_t-\xi_0)\big]=0$. Define then $\psi_{t,\tau}(\bs\theta,\xi)=\phi_\tau\{\eta_t(\bs\theta)-\xi\}$, $\widetilde\psi_{t,\tau}(\bs\theta,\xi)=\phi_\tau\{\widetilde\eta_t(\bs\theta)-\xi\}$, $Q_n(\xi,\bs\theta)=\frac{1}{n}\sum_{t=1}^n\psi_{t,\tau}(\bs\theta,\xi)$, and  $\widetilde{Q}_n(\xi,\bs\theta)=\frac{1}{n}\sum_{t=1}^n\widetilde\psi_{t,\tau}(\bs\theta,\xi)$ for any generic $(\bs\theta,\xi)$. At the truth, set $\psi_t=\psi_{t,\tau}(\bs\theta_0,\xi_0)=\phi_\tau(\eta_t-\xi_0)$, $w_t=w_\tau(\eta_t-\xi_0)$, and $\Psi_0=\E[w_t]=\tau\{1-F_\eta(\xi_0)\}+(1-\tau)F_\eta(\xi_0)\ge\min\{\tau,1-\tau\}>0$. As such, $\phi_\tau$ denotes the deterministic piecewise-linear map, $\psi_{t,\tau}(\bs\theta,\xi)$ denotes the observation-specific estimating function, and $\psi_t$ denotes the estimating function evaluated at $(\bs\theta_0,\xi_0)$. Finally, we will henceforth denote by $\varphi_t$ the influence function.

The QML criterion is $\widetilde G_n(\bs\theta)=\frac{1}{n}\sum_{t=1}^n g\big(y_t,\widetilde\sigma_t(\bs\theta)\big)$, where $g(y,s)=\log\big(s^{-1}h(y/s)\big)$. To avoid ambiguity in the derivatives of the scale likelihood, define
\[
 a(u)=\left.\frac{\partial}{\partial s}g(u,s)\right|_{s=1}, \qquad b(u)=\left.\frac{\partial^2}{\partial s^2}g(u,s)\right|_{s=1}.
\]
By scale equivariance, $\partial_{\bs\theta}g\big(y_t,\sigma_t(\bs\theta_0)\big)=a(\eta_t)\bs{D}_t=:\bs{s}_t$, whereas $\bs{H}_0=-\E\left[\partial_{\bs\theta\bs\theta'}^2g\big(y_t,\sigma_t(\bs\theta_0)\big)\right]$ is the negative population Hessian, which is positive definite under the following conditions.

\renewcommand{\theassumption}{A\arabic{assumption}}
\begin{assumption}[Innovations and target expectile]\label{ass:innovation_process}
The innovations $\{\eta_t\}$ are iid and independent of $\mathcal{F}_{t-1}$, with $\E\eta_t=0$ and $\E\eta_t^2=1$.  For the fixed level $\tau\in(0,1)$, $\E\{\phi_\tau(\eta_t-\xi)\}=0$ has a unique solution $\xi_0$, and $F_\eta$ is continuous at $\xi_0$.
\end{assumption}

We require continuity at $\xi_0$ such that $\Pr(\eta_t=\xi_0)=0$ for the local piecewise-linear expansion. It is also a standard regularity condition for the Gaussian limit of sample expectiles \cite{holzmann2016expectile}. In particular, we do not need to impose absolute continuity or to bound the density away from zero.

\begin{assumption}[Stationary causal scale model]\label{ass:stationarity_mixing}
The process $\{y_t\}$ is a strictly stationary and ergodic causal solution of \eqref{eq:model_appendix}.  For every $\bs\theta$ in a neighborhood of $\bs\theta_0$, $\sigma_t(\bs\theta)$ and its first two parameter derivatives are $\mathcal{F}_{t-1}$-measurable.
\end{assumption}

\begin{assumption}[Parameter space]\label{ass:compactness}
The parameter space $\Theta\subset\mathbb{R}^m$ is compact and $\bs\theta_0\in\Theta$. For any result using the QML first-order condition, $\bs\theta_0$ is in the interior of $\Theta$.
\end{assumption}

\begin{assumption}[Conditional scale and identification]\label{ass:volatility_process}
For every real sequence $(x_j)_{j\ge1}$, the map $\bs\theta\mapsto\sigma(x_1,x_2,\ldots;\bs\theta)$ is continuous on $\Theta$ and twice continuously differentiable on an open neighborhood $V_0$ of $\bs\theta_0$. There is a constant $\underline\sigma>0$ such that $\inf_{\bs\theta\in\Theta}\sigma_t(\bs\theta)\ge\underline\sigma$ and $\inf_{\bs\theta\in\Theta}\widetilde\sigma_t(\bs\theta)\ge\underline\sigma$ almost surely. It also holds not only that $\sigma_t(\bs\theta_0)=\sigma_t(\bs\theta)$ almost surely if and only if $\bs\theta=\bs\theta_0$, but also that $\bs{I}_0=\E(\bs{D}_t\bs{D}_t')$ is finite and positive definite.
\end{assumption}

\begin{assumption}[Finite-history approximation]\label{ass:approximation}
There exist $\rho\in(0,1)$ and a nonnegative random vector $(C_0,C_1,C_2)$ such that $\sup_{\bs\theta\in\Theta}|\widetilde\sigma_t(\bs\theta)-\sigma_t(\bs\theta)|\le C_0\rho^t$, $\sup_{\bs\theta\in V_0} \left\|\partial_{\bs\theta}\widetilde\sigma_t(\bs\theta)-\partial_{\bs\theta}\sigma_t(\bs\theta)\right\|\le C_1\rho^t$, and $\sup_{\bs\theta\in V_0}\left\|\partial^2_{\bs\theta\bs\theta'}\widetilde\sigma_t(\bs\theta)-\partial^2_{\bs\theta\bs\theta'} \sigma_t(\bs\theta)\right\|\le C_2\rho^t$ almost surely. The initialization envelopes are such that $\E[C_0^4+C_1^4+C_2^4]<\infty$.
\end{assumption}

We necessitate the fourth-moment condition in Assumption~\ref{ass:approximation} only for the QML and plug-in covariance estimation. The expectile consistency and finite-history score approximation require only $\E[C_0^2]<\infty$.

\begin{assumption}[Moment and derivative envelopes of order $r$]
\label{ass:moments}
For $r\ge2$, $\E|\eta_t|^r+\E|y_t|^r<\infty$. Letting $A_t(\bs\theta)=\sigma_t(\bs\theta_0)/\sigma_t(\bs\theta)$ for some neighborhood $V\subset V_0$ of $\bs\theta_0$, there exists a stationary, ergodic, predictable process $B_t\ge1$ such that
\begin{align*}
 \sup_{\bs\theta\in V}
 \big\{ |A_t(\bs\theta)| +\|\partial_{\bs\theta}A_t(\bs\theta)\| +\|\partial^2_{\bs\theta\bs\theta'}A_t(\bs\theta)\| \big\} &\le B_t,\\
 \sup_{\bs\theta\in V} \big\{ \|\bs{D}_t(\bs\theta)\| +\|\partial_{\bs\theta}\bs{D}_t(\bs\theta)\| \big\} &\le B_t,
\end{align*}
and $\E B_t^r<\infty$.
\end{assumption}

Letting Assumption~\ref{ass:moments} hold with $r=4$ essentially requires the fourth moments of $(\eta_t,y_t)$ and of the local derivative envelopes. Similarly, stating that Assumption~\ref{ass:moments} holds with $r=2$ only calls for second moments.

\begin{assumption}[Instrumental-likelihood identification]
\label{ass:identification_instrumental}
$\E[g(\eta_t,s)]<\E[g(\eta_t,1)]$ for every nonnegative $s\ne1$.
\end{assumption}

\begin{assumption}[Instrumental-likelihood smoothness]
\label{ass:differentiability_instrumental}
The instrumental density $h$ is strictly positive and three times continuously differentiable.  There exist $\nu_h\ge0$ and $C_h<\infty$ such that the first three scale derivatives of $g(u,s)=\log\big(s^{-1}h(u/s)\big)$ have polynomial growth of order $\nu_h$ on compact subsets of $(0,\infty)$.  In particular, $\E|\eta_t|^{2\nu_h}<\infty$ and
\[
 \left|u\,\frac{h'(u)}{h(u)}\right|+\left|u^2\left\{\frac{h'(u)}{h(u)}\right\}'\right|\le C_h(1+|u|^{\nu_h}).
\]
Moreover, $\E[a(\eta_t)]=0$, $\kappa_h:=-\E[b(\eta_t)]>0$, and $\E|b(\eta_t)|<\infty$.
\end{assumption}

The first two growth restrictions correspond to the smooth version of Conditions A4 and A9 in \cn{francq2015risk}, whereas we use the third scale derivative to derive the equivalence of the finite-history QML Hessian.

\begin{assumption}[Fourth-order QML moments]\label{ass:qml_moments}
The QML score and curvature are such that $\E|a(\eta_t)|^4<\infty$ and $\E|b(\eta_t)|^2<\infty$. For the neighborhood $V$ in Assumption~\ref{ass:moments}, let $\mathcal{I}_t(\bs\theta)$ denote the closed interval with endpoints $\sigma_t(\bs\theta)$ and $\widetilde\sigma_t(\bs\theta)$. There exist stationary nonnegative envelopes $(G_{1,t},G_{2,t},G_{3,t})$ satisfying $\E[G_{1,t}^4+G_{2,t}^4+G_{3,t}^2]<\infty$ and $\sup_{\bs\theta\in V}\sup_{s\in\mathcal{I}_t(\bs\theta)}\left|\partial_s^j g(y_t,s)\right|\le G_{j,t}$ almost surely, for $j=1,2,3$.
\end{assumption}

Assumptions~\ref{ass:qml_moments} and~\ref{ass:moments} with $r=4$ supply the local fourth-order derivative conditions for the QMLE linear representation and for the joint martingale central limit theorem. For Gaussian QML inference, these requirements reduce to polynomial eighth-moment conditions on the innovations and fourth-moment on $y_t$.

\begin{assumption}[Plug-in variance estimation]
\label{ass:plugin}
There exist stationary and ergodic nonnegative envelopes $L_{a,t}$ and $L_{b,t}$, with $\E[L_{a,t}^4]<\infty$ and $\E[L_{b,t}^2]<\infty$, such that, for all $\bs\theta_1,\bs\theta_2\in V$, $|a(A_t(\bs\theta_1)\eta_t)-a(A_t(\bs\theta_2)\eta_t)|\le L_{a,t}\|\bs\theta_1-\bs\theta_2\|$ and $|b(A_t(\bs\theta_1)\eta_t)-b(A_t(\bs\theta_2)\eta_t)|\le L_{b,t}\|\bs\theta_1-\bs\theta_2\|$.
\end{assumption}

Assumption~\ref{ass:plugin} is not necessary for the consistency and asymptotic normality of the two-step expectile estimators. However, it is key for the consistency of the plug-in estimator of the asymptotic variance and for the construction of asymptotically-valid confidence intervals. A more primitive alternative is to impose polynomial-growth bounds on the derivatives of $a$ and $b$. Assumption~\ref{ass:moments} with $r=8$ actually ensures that such bounds hold in the context of Gaussian QML estimation given that $a(u)=u^2-1$ and $b(u)=1-3u^2$.

\begin{assumption}[Observed forecast state]\label{ass:forecast}
For $\bs\theta\in V$, the volatility recursion and the first two derivative recursions admit the augmented Markov representation $\mathcal{U}_t(\bs\theta)=\Phi_{\bs\theta}\{\mathcal{U}_{t-1}(\bs\theta),\eta_t\}$. For any two initial states $u$ and $v$, let $\mathcal{U}_t^u(\bs\theta)$ and $\mathcal{U}_t^v(\bs\theta)$ denote the solutions driven by the same sequence of innovations. There exist $\rho\in(0,1)$ and a stationary nonnegative random sequence $C_t$ with $\E[C_t^q]<\infty$ and $q>2$ such that
\begin{equation}\label{eq:forecast_augmented_contraction}
 \sup_{\bs\theta\in V}\left\|\mathcal{U}_t^u(\bs\theta)-\mathcal{U}_t^v(\bs\theta)\right\|\le C_t\rho^t\big(1+\|u\|+\|v\|\big)
\end{equation}
almost surely. The maps from $\mathcal{U}_t(\bs\theta)$ to $\sigma_{t+1}(\bs\theta)$ and its first two parameter derivatives are locally Lipschitz, uniformly on $V$, with envelopes having finite $q$-th moments.
\end{assumption}

The uniform contraction in \eqref{eq:forecast_augmented_contraction} implies the finite-history bounds in Assumption~\ref{ass:approximation}. Imposing a contraction only at $\bs\theta_0$ would not suffice since we must also control the parameter derivatives uniformly on $V$. Moreover, if we initialize the feasible recursion at a state with finite $q$-th moment, restarting the augmented recursion at time $n-m$ for any $m\ge1$ from a fixed state would produce $\sigma_{n+1}^{[m]}$ and $D_{n+1}^{[m]}$ that are measurable with respect to $\mathcal{H}_{n,m}=\sigma(\eta_{n-m+1},\ldots,\eta_n)$ and such that
\begin{equation}\label{eq:forecast_terminal_contraction}
 \left|\sigma_{n+1}(\bs\theta_0)-\sigma_{n+1}^{[m]}\right|+\left\|\bs{D}_{n+1}-\bs{D}_{n+1}^{[m]}\right\|\le C_{[m],t}\,\rho^m,
 \quad\textrm{with}~~\sup_n\E(C_{[m],t}^2)<\infty
\end{equation}
for some stationary nonnegative random sequence $C_{[m],t}$. Finally, it also follows that
\[
 \sigma_{n+1}(\bs\theta_0)\left\{1+\|\bs{D}_{n+1}\|\right\}+\sup_{\bs\theta\in V}\left\|\partial_{\bs\theta\bs\theta'}^2\sigma_{n+1}(\bs\theta)\right\|=O_p(1).
\]

\begin{remark}
Although we could consider weaker short-memory conditions, the geometric contraction implied by Assumption~\ref{ass:forecast} is standard in the literature \ca{francq2015risk}{see, for instance}, apart from being convenient for the forward initialization and terminal-state approximation bounds in Theorem~\ref{thm:conditional_expectiles}. The blocking argument remains valid if, for every $m\ge1$, there exist $\sigma_{n+1}^{[m]}$ and $\bs{D}_{n+1}^{[m]}$ measurable with respect to $\sigma(\eta_{n-m+1},\ldots,\eta_n)$ such that
\[
 \delta_m:=\sup_{n\ge1}\left\|\sigma_{n+1}(\bs\theta_0)-\sigma_{n+1}^{[m]}\right\|_2+\sup_{n\ge1}\left\|\bs{D}_{n+1}-\bs{D}_{n+1}^{[m]}\right\|_2\longrightarrow0.
\]
It would then suffice to choose a deterministic sequence $m_n$ such that $m_n\to\infty$, $m_n/n\to0$ and $\delta_{m_n}\to0$. The feasible forecast recursion must also satisfy $\sqrt{n}\,\sup_{\bs\theta\in V}\left|\widetilde\sigma_{n+1}(\bs\theta)-\sigma_{n+1}(\bs\theta)\right|=o_p(1)$. For consistency of the feasible estimator of the conditional variance, it is also sufficient to impose $\sup_{\bs\theta\in V}\left\|\partial_{\bs\theta} \widetilde\sigma_{n+1}(\bs\theta)-\partial_{\bs\theta}\sigma_{n+1}(\bs\theta)\right\|=o_p(1)$. Under these conditions, the last $m_n$ score contributions are $o_p(1)$ given that $m_n/n\to0$, while the terminal state of the finite history does not depend on the sum of earlier scores given that innovations are independent over time. In addition, the $L^2$ approximation and tightness of the normalized score vector ensure that replacing the exact terminal coefficient by its finite-memory counterpart has a contribution of order $o_p(1)$. Geometric contraction not only yields these conditions with $\delta_m=O(\rho^m)$, but also holds for the stationary GARCH and GJR-GARCH processes we consider in Proposition~\ref{prop:gjr_verification}. The same argument may also apply to more general causal nonlinear volatility processes that admit an innovation-based finite-memory approximation, even at subgeometric rates.
\end{remark}

\begin{remark}
$\bs{H}_0=\kappa_h\bs{I}_0$ under Assumptions~\ref{ass:volatility_process}, \ref{ass:moments} with $r=2$ and \ref{ass:differentiability_instrumental}. Indeed, differentiating $\partial_{\bs\theta}g\big(y_t,\sigma_t(\bs\theta_0)\big)=a(\eta_t)D_t$ yields $\partial^2_{\bs\theta\bs\theta'} g\big(y_t,\sigma_t(\bs\theta_0)\big)=\big[a(\eta_t)+b(\eta_t)\big]\bs{D}_t\bs{D}_t'+a(\eta_t)\partial_{\bs\theta}\bs{D}_t'$.
However, it follows from $\eta_t$ independent of $\mathcal{F}_{t-1}$ and $\E[a(\eta_t)]=0$ that
\[
 \bs{H}_0=-\E\left[\partial^2_{\bs\theta\bs\theta'}g\big(y_t,\sigma_t(\bs\theta_0)\big)\right]=-\E[b(\eta_t)]\E[\bs{D}_t\bs{D}_t']=\kappa_h\bs{I}_0.
\]
\end{remark}

\begin{proposition}[Stationary GJR-GARCH processes]
\label{prop:gjr_verification}
Let $y_t=\sigma_t(\bs\theta_0)\eta_t$, with
$$
\sigma_t^2(\bs\theta)=\omega+\sum_{i=1}^q\big[\alpha_i+\gamma_i\bs{1}(y_{t-i}<0)\big]\,y_{t-i}^2 +\sum_{j=1}^p\beta_j\sigma_{t-j}^2(\bs\theta).
$$
Suppose that $\bs\theta_0$ in the interior of the compact parameter space $\Theta$ and that $\omega\ge\underline\omega>0$, $\alpha_i\ge0$, $\alpha_i+\gamma_i\ge0$, $\beta_j\ge0$, and $\sum_{j=1}^p\beta_j\le\bar\beta<1$ uniformly on $\Theta$. Let $\mathcal{A}_t(\bs\theta_0)$ denote the nonnegative companion matrix of the variance recurrence. Assume that it has a negative top Lyapunov exponent and that $\E\|\mathcal A_t(\bs\theta_0)\|^{r/2}<1$. Finally, suppose that the positive- and negative-shock ARCH and GARCH polynomials satisfy the usual no-common-root and nonredundancy conditions, and that the innovation law does not concentrate on two points. It then follows that Assumptions~\ref{ass:stationarity_mixing} to~\ref{ass:moments} hold, and that the QML estimator is consistent if the instrumental likelihood meets the conditions in Assumptions~\ref{ass:identification_instrumental} and~\ref{ass:differentiability_instrumental}. Assumptions~\ref{ass:qml_moments} and~\ref{ass:plugin} automatically hold if we set $r=8$ in the Gaussian QML case, whereas Assumption~\ref{ass:forecast} holds for the observed stationary forecast state.
\end{proposition}
\begin{proof}
The negative top Lyapunov exponent ensures there is a unique strictly stationary causal solution of the variance recurrence, with the matrix-moment contraction yielding the corresponding $r$-th moments \ca{bougerol1992stationarity,ling2002necessary}{see, among others,}. For the GJR-GARCH$(1,1)$ process, the latter condition reduces to
\[
 \E\left\{\left[\beta_0+\big(\alpha_0+\gamma_0\bs{1}(\eta_t<0)\big)\eta_t^2\right]^{r/2}\right\}<1.
\]
The lower bound follows from $\sigma_t^2(\bs\theta)\ge\underline\omega$, whereas the standard polynomial nonredundancy conditions imply global identification and positive definiteness of $I_0$. Because the GARCH polynomial is uniformly stable, the infinite-ARCH coefficients and their first two derivatives decay at geometric rates. As such, replacing the unobserved pre-sample values by arbitrary fixed values produces the bounds in Assumption~\ref{ass:approximation}. Given the matrix-moment contraction, the same stable recursions for the (normalized) first and second parameter derivatives yield the envelope in Assumption~\ref{ass:moments}. These are the GJR-GARCH counterparts of Conditions A2, A5, A7, A8, and A10 in \cn{francq2015risk}. In turn, the plug-in Lipschitz bounds are polynomial for the Gaussian instrumenatl likelihood because $a(u)=u^2-1$ and $b(u)=1-3u^2$. As a result, it suffices to impose the existence of fourth moments in Assumption~\ref{ass:moments}.  Finally, stationarity ensures forecast-state tightness, while the above geometric forgetting argument gives the pointwise initialization bound in Assumption~\ref{ass:forecast}.
\end{proof}

\subsection{Auxiliary results}

We enunciate two intermediate results. The first documents two elementary properties that replace ordinary differentiability of the expectile score, whereas the second gives a uniform approximation result that we apply to the expectile score, to the QML score and Hessian, and to some components of the asymptotic variance estimator.

\begin{lemma}[Piecewise-linear expectile identity]\label{lem:expectile_identity}
Let $\bar\tau=\max\{\tau,1-\tau\}$ and $c_\tau=|1-2\tau|$. For every $u,v,h\in\mathbb{R}$, $|\phi_\tau(u)-\phi_\tau(v)|\le\bar\tau|u-v|$ and $\left|\phi_\tau(u+h)-\phi_\tau(u)-w_\tau(u)h\right|\le c_\tau|h|\bs{1}(|u|\le|h|)$.
\end{lemma}
\begin{proof}
Because
\[
 \phi_\tau(u)= \begin{cases}
  (1-\tau)u & u<0\\
  \tau u    & u\ge0,
 \end{cases}
\]
the slopes of $\phi_\tau$ are bounded in absolute value by $\bar\tau$. As such, the first inequality result follows from the mean-value argument on each half-line given the continuity at zero. To establish the second inequality, we have to consider four cases. If $u<0$ and $u+h<0$, then
\[
 \phi_\tau(u+h)-\phi_\tau(u)-w_\tau(u)h=(1-\tau)(u+h)-(1-\tau)u-(1-\tau)h=0.
\]
In fact, the same applies if $u\ge0$ and $u+h\ge0$. Suppose next that $u<0\le u+h$. Since $w_\tau(u)=1-\tau$,
\begin{align*}
 \phi_\tau(u+h)-\phi_\tau(u)-w_\tau(u)h 
    &=\tau(u+h)-(1-\tau)u-(1-\tau)h\\
    &=(2\tau-1)(u+h).
\end{align*}
The crossing event is $\{u<0\le u+h\}\cup\{u+h<0\le u\}$, which implies $0\le u+h\le h=|h|$ and $|u|\le|h|$. Altogether, we obtain $|(2\tau-1)(u+h)|\le c_\tau|h|\bs{1}(|u|\le|h|)$. Finally, if $u\ge0>u+h$, then $h<0$, $w_\tau(u)=\tau$, and $\phi_\tau(u+h)-\phi_\tau(u)-w_\tau(u)h=(1-2\tau)(u+h)$. However, since $h\le u+h<0$, it follows that $|u+h|\le|h|$ just as in the crossing event, completing the proof.
\end{proof}

\begin{lemma}[Uniform finite-history approximation]\label{lem:approx}
Assume $V\subset\Theta$ and $\mathcal{B}$ are compact. Let $\bs{Z}_t(\bs\theta)$ collect any subset of $\sigma_t(\bs\theta)$ and its first two derivatives, with $\widetilde{\bs{Z}}_t(\bs\theta)$ denoting its finite-history counterpart. Suppose that $\sup_{\bs\theta\in V}\|\bs{Z}_t(\bs\theta)-\widetilde{\bs{Z}}_t(\bs\theta)\|\le C_Z\rho^t$ for some nonnegative random variable $C_Z$, and that, for a measurable map $\bs{q}_t(z,b)$, there is a nonnegative random variable $M_t$ such that
\begin{equation}\label{eq:q_lipschitz_initialization}
 \sup_{\bs\theta\in V,\,b\in\mathcal{B}}\left\|\bs{q}_t\big(\bs{Z}_t(\bs\theta),b\big)-\bs{q}_t\big(\widetilde{\bs{Z}}_t(\bs\theta),b\big)\right\|\le C_ZM_t\rho^t,
\end{equation}
with $\sup_{t\ge1}\E(C_ZM_t)<\infty$. It then follows almost surely that
\begin{equation}
 \sup_{\bs\theta\in V,\,b\in\mathcal{B}}\left\|\frac{1}{n}\sum_{t=1}^n\left[\bs{q}_t\big(\bs{Z}_t(\bs\theta),b\big)-\bs{q}_t\big(\widetilde{\bs{Z}}_t(\bs\theta),b\big)\right]\right\|=O(1/n),\label{eq:uniform_approx_general}
\end{equation}
In particular, for every compact $K\subset\mathbb{R}$,
\begin{equation}\label{eq:uniform_approx_expectile}
 \sup_{\bs\theta\in V,\,\xi\in K} \sqrt{n}\,|Q_n(\xi,\bs\theta) - \widetilde Q_n(\xi,\bs\theta)| \longrightarrow0\qquad a.s.
\end{equation}
\end{lemma}
\begin{proof}
Let $\Delta_t=\sup_{\bs\theta\in V,\,b\in\mathcal{B}} \left\|\bs{q}_t\big(\bs{Z}_t(\bs\theta),b\big)-\bs{q}_t\big(\widetilde{\bs{Z}}_t(\bs\theta),b\big)\right\|$ and $K_\Delta:=\sum_{t=1}^{\infty}\Delta_t$. It follows from \eqref{eq:q_lipschitz_initialization} that
\[
 \E(K_\Delta)\le\sum_{t=1}^{\infty}\rho^t\E(C_ZM_t)\le\sup_{t\ge1}\E(C_ZM_t)\,\frac{\rho}{1-\rho}<\infty
\]
by Tonelli's theorem. In turn, this implies that $K_\Delta$ is almost surely finite (since it is a nonnegative random variable with finite mean). Accordingly, for every $n$,
\[
 \sup_{\bs\theta,b}\left\|\frac{1}{n}\sum_{t=1}^n[\bs{q}_t\big(\bs{Z}_t(\bs\theta),b\big)-\bs{q}_t\big(\widetilde{\bs{Z}}_t(\bs\theta),b\big)]\right\|\le \frac{K_\Delta}{n},
\]
establishing \eqref{eq:uniform_approx_general}. In the expectile application, we consider $\bs{q}_t(z,\xi)=\phi_\tau(y_t/z-\xi)$. The lower bound on the scales and Lemma~\ref{lem:expectile_identity} imply
\begin{align*}
 \sup_{\bs\theta\in V,\,\xi\in K} |\widetilde\psi_{t,\tau}(\bs\theta,\xi) -\psi_{t,\tau}(\bs\theta,\xi)|&\quad\le \bar\tau |y_t| \sup_{\bs\theta\in V} \left|\frac{1}{\widetilde\sigma_t(\bs\theta)} -\frac{1}{\sigma_t(\bs\theta)}\right|\\
 &\quad= \bar\tau |y_t| \sup_{\bs\theta\in V} \frac{|\widetilde\sigma_t(\bs\theta) -\sigma_t(\bs\theta)|}  {\widetilde\sigma_t(\bs\theta)\sigma_t(\bs\theta)}\le\bar\tau\underline\sigma^{-2}|y_t|C_0\rho^t.
\end{align*}
Moreover, $\E(C_0|y_t|)\le(\E[C_0^2])^{1/2}(\E[y_t^2])^{1/2}<\infty$ by Cauchy-Schwarz under Assumptions~\ref{ass:stationarity_mixing} and \ref{ass:approximation}. This ensures that \eqref{eq:q_lipschitz_initialization} holds with $C_Z=C_0$ and $M_t=\bar\tau\underline\sigma^{-2}|y_t|$, thus proving \eqref{eq:uniform_approx_expectile}. For the QML criterion and score, the chain rule and the mean-value theorem imply not only that $|g\{y_t,\widetilde\sigma_t(\bs\theta)\}-g\{y_t,\sigma_t(\bs\theta)\}|\le G_{1,t}C_0\rho^t$, but also that $\left\|\partial_{\bs\theta}g\{y_t,\widetilde\sigma_t(\bs\theta)\}-\partial_{\bs\theta}g\{y_t,\sigma_t(\bs\theta)\}\right\|\le M_{1,t}\rho^t$ uniformly on $V$, where $M_{1,t}$ is a finite sum of products of $G_{1,t}$, $G_{2,t}$, $B_t$, $C_0$ and $C_1$. The analogous Hessian decomposition uses $G_{1,t}$, $G_{2,t}$, $G_{3,t}$, $B_t$, $C_0$, $C_1$ and $C_2$. The moment requirements in Assumptions~\ref{ass:approximation}, \ref{ass:moments}, and \ref{ass:qml_moments} imply $\sup_t\E M_{j,t}<\infty$ by Hölder's inequality. Accordingly, not only the exact and finite-history QML scores are equivalent at the root-$n$ order, but also their average Hessians are asymptotically equivalent.
\end{proof}

The next intermediate result controls the nonsmooth remainder on a neighborhood of order $n^{-1/2}$. It is the main technical device we use to establish the Bahadur representation.

\begin{lemma}[Uniform local expansion of the expectile score]
\label{lem:local_expectile_expansion}
For $M<\infty$, let
\[
 \mathcal{N}_n(M)=\left\{(\bs\theta,\xi):~\sqrt{n}\|\bs\theta-\bs\theta_0\|\le M,~\sqrt{n}|\xi-\xi_0|\le M\right\}.
\]
Under Assumptions~\ref{ass:innovation_process}, \ref{ass:stationarity_mixing}, \ref{ass:volatility_process}, and \ref{ass:moments} with $r=2$,
\begin{align}
 &\sup_{(\bs\theta,\xi)\in\mathcal{N}_n(M)}\Bigg|\sqrt{n}\big[Q_n(\xi,\bs\theta)-Q_n(\xi_0,\bs\theta_0)\big]\notag\\[-0.2cm]
 &\hspace{2.4cm}+\left(\frac{1}{n}\sum_{t=1}^n w_t\eta_t\bs{D}_t'\right)\sqrt{n}(\bs\theta-\bs\theta_0)+\left(\frac{1}{n}\sum_{t=1}^n w_t\right)\sqrt{n}(\xi-\xi_0)\Bigg|=o_p(1)
 \label{eq:uniform_local_expansion}
\end{align}
and $\sqrt{n}\big[Q_n(\xi,\bs\theta)-Q_n(\xi_0,\bs\theta_0)\big]=-\xi_0\Psi_0\bs{J}_0'\sqrt{n}(\bs\theta-\bs\theta_0)-\Psi_0\sqrt{n}(\xi-\xi_0)+o_p(1)$ uniformly on $\mathcal{N}_n(M)$.
\end{lemma}
\begin{proof}
We proceed in four steps. The first step provides a uniform expansion for the scale ratio. The second step derives an exact decomposition of the nonsmooth score, whereas the third step establishes uniform control of both smooth and nonsmooth remainders. Finally, the last step shows that we can replace sample coefficients by their limits.\vskip 1em

\noindent\emph{Step 1}\\
We let $A_t(\bs\theta)=\sigma_t(\bs\theta_0)/\sigma_t(\bs\theta)$, $r_t=\eta_t-\xi_0$, and $h_t(\bs\theta,\xi)=\big[A_t(\bs\theta)-1\big]\eta_t-(\xi-\xi_0)$. Assumption~\ref{ass:volatility_process} dictates that $A_t(\bs\theta)$ is twice continuously differentiable on $V$, with first derivative at the true parameter amounting to $\partial_{\bs\theta}A_t(\bs\theta_0)=-\partial_{\bs\theta}\log\sigma_t(\bs\theta_0)=-\bs{D}_t$. It then follows from Taylor's theorem and the Hessian envelope in Assumption~\ref{ass:moments} that
\begin{align}
 A_t(\bs\theta)-1&=-\bs{D}_t'(\bs\theta-\bs\theta_0)+R_{A,t}(\bs\theta), \label{eq:A_local_taylor}\\
 |R_{A,t}(\bs\theta)|&\le C_MB_t\|\bs\theta-\bs\theta_0\|^2\le C_MB_t/n. \label{eq:A_remainder_bound}
\end{align}
uniformly for $\sqrt{n}\|\bs\theta-\bs\theta_0\|\le M$. In turn, the first-derivative envelope in Assumption~\ref{ass:moments} implies
\begin{equation}\label{eq:h_uniform_bound}
 \sup_{(\bs\theta,\xi)\in\mathcal{N}_n(M)} |h_t(\bs\theta,\xi)| \le \frac{C_M}{\sqrt{n}}(1+B_t|\eta_t|)
\end{equation}
for $n$ large enough.\vskip 1em

\noindent\emph{Step 2}\\
Since $\psi_{t,\tau}(\bs\theta,\xi)=\phi_\tau\big(r_t+h_t(\bs\theta,\xi)\big)$, Lemma~\ref{lem:expectile_identity} yields $\phi_\tau(r_t+h_t)-\phi_\tau(r_t)=w_t h_t+\rho_{t,n}(\bs\theta,\xi)$ with
\begin{equation}\label{eq:kink_remainder_bound}
 |\rho_{t,n}(\bs\theta,\xi)| \le c_\tau |h_t(\bs\theta,\xi)| \bs{1}\{|r_t|\le |h_t(\bs\theta,\xi)|\}.
\end{equation}
Plugging \eqref{eq:A_local_taylor} into $w_t h_t$ gives the algebraic identity
\begin{equation}
 \psi_{t,\tau}(\bs\theta,\xi)-\psi_t =-w_t\eta_tD_t'(\bs\theta-\bs\theta_0)-w_t(\xi-\xi_0)   +w_t\eta_tR_{A,t}(\bs\theta)+\rho_{t,n}(\bs\theta,\xi).\label{eq:score_local_decomp}
\end{equation}\vskip 1em

\noindent\emph{Step 3}\\
For the smooth Taylor remainder, \eqref{eq:A_remainder_bound} implies that
\begin{equation}\label{eq:inequality}
 \sup_{\mathcal{N}_n(M)}\sqrt{n}\left|\frac{1}{n}\sum_{t=1}^n w_t\eta_tR_{A,t}(\bs\theta)\right|\le\frac{C_M}{\sqrt{n}}\frac{1}{n}\sum_{t=1}^n B_t|\eta_t|,
\end{equation}
where the sequence $B_t|\eta_t|$ is integrable by Hölder's inequality under Assumptions~\ref{ass:innovation_process} and \ref{ass:moments}. It is also stationary and ergodic because it is a measurable function of the stationary ergodic process and its past. It then follows from the ergodic theorem that $\frac{1}{n}\sum_tB_t|\eta_t|=O_p(1)$, thus ensuring that the right-hand side of the inequality \eqref{eq:inequality} is $o_p(1)$. As for the kink remainder, combining \eqref{eq:h_uniform_bound} and \eqref{eq:kink_remainder_bound} gives way to
\begin{equation}
 \sup_{\mathcal{N}_n(M)}\sqrt{n}\left|\frac{1}{n}\sum_{t=1}^n\rho_{t,n}(\bs\theta,\xi)\right|\le\frac{C_M}{n}\sum_{t=1}^n(1+B_t|\eta_t|)\bs{1}\left[|\eta_t-\xi_0|\le\frac{C_M}{\sqrt{n}}(1+B_t|\eta_t|)\right].
 \label{eq:kink_average_bound}
\end{equation}
For each fixed realization, the indicator on the right-hand side of \eqref{eq:kink_average_bound} converges to zero as long as $\eta_t\ne\xi_0$. However, the latter occurs with probability one under the continuity condition in Assumption~\ref{ass:innovation_process}. In turn, $B_t|\eta_t|$ is integrable by Hölder's inequality and independence of $\eta_t$ from $\mathcal{F}_{t-1}$. Dominated convergence therefore yields
\[
 \E\left\{\left(1+B_t|\eta_t|\right)\bs{1}\left[|\eta_t-\xi_0|\le\frac{C_M}{\sqrt{n}}(1+B_t|\eta_t|)\right]\right\}\longrightarrow0.
\]
Taking expectations in \eqref{eq:kink_average_bound} and applying Markov's inequality establish that the kink term is $o_p(1)$ uniformly on $\mathcal{N}_n(M)$.  The uniform local expansion in \eqref{eq:uniform_local_expansion} then follows by summing \eqref{eq:score_local_decomp}.\vskip 1em

\noindent\emph{Step 4}\\
The variables $w_t$ and $w_t\eta_t\bs{D}_t$ are measurable functions of the stationary ergodic process. They are also integrable because $w_t$ is bounded and $\E\|w_t\eta_t\bs{D}_t\|\le\bar\tau\,\E(|\eta_t|B_t)<\infty$. By the ergodic theorem, $\frac{1}{n}\sum_{t=1}^n w_t\to\Psi_0$ and $\frac{1}{n}\sum_{t=1}^n w_t\eta_t\bs{D}_t\to\E(w_t\eta_t\bs{D}_t)$ almost surely. Since $\bs{D}_t$ is $\mathcal{F}_{t-1}$-measurable and $\eta_t$ is independent of $\mathcal{F}_{t-1}$, $\E(w_t\eta_t\bs{D}_t)=\E(w_t\eta_t)\E(\bs{D}_t)$. However, the expectile first-order condition implies that $\E\{w_t(\eta_t-\xi_0)\}=\E(w_t\eta_t)-\xi_0\E(w_t)$ must equal zero, so that $\E(w_t\eta_t)=\xi_0\Psi_0$ and $\E(w_t\eta_t\bs{D}_t)=\xi_0\Psi_0\bs{J}_0$. Finally, replacing sample coefficients by their limits does not affect uniformity in that
\[
 \sup_{\mathcal{N}_n(M)}\left|\left(\frac{1}{n}\sum_{t=1}^n w_t\eta_t\bs{D}_t'-\xi_0\Psi_0\bs{J}_0'\right)\sqrt{n}(\bs\theta-\bs\theta_0)\right|
 \le M\left\|\frac{1}{n}\sum_{t=1}^n w_t\eta_t\bs{D}_t-\xi_0\Psi_0\bs{J}_0\right\|=o_p(1)
\]
and, similarly,
\[
 \sup_{\mathcal{N}_n(M)}\left|\left(\frac{1}{n}\sum_{t=1}^n w_t-\Psi_0\right)\sqrt{n}(\xi-\xi_0)\right|\le M\left|\frac{1}{n}\sum_{t=1}^n w_t-\Psi_0\right|=o_p(1),
\]
thereby completing the proof.
\end{proof}

\subsection{First-step QML estimation}

\begin{proposition}[QMLE consistency and linear representation]
\label{prop:consistency_qmle}
If Assumptions~\ref{ass:innovation_process} to~\ref{ass:differentiability_instrumental} hold with $r=2$ in Assumption~\ref{ass:moments}, then $\widehat{\bs\theta}_n\asconv\bs\theta_0$. If we further assume that $\bs\theta_0$ is in the interior of $\Theta$ and Assumptions~\ref{ass:moments} with $r=4$ and~\ref{ass:qml_moments} hold, then
\begin{equation}\label{eq:qMLE_linear_rep}
 \sqrt{n}(\widehat{\bs\theta}_n-\bs\theta_0)=\bs{H}_0^{-1}\frac{1}{\sqrt{n}}\sum_{t=1}^n\bs{s}_t+o_p(1),
\end{equation}
where $\bs{s}_t=a(\eta_t)\bs{D}_t$.
\end{proposition}
\begin{proof}
Consistency follows readily from Theorem 1 of \cn{francq2015risk}, with stationarity and ergodicity stemming from Assumption~\ref{ass:stationarity_mixing}; compactness, continuity, positivity and scale identification from Assumptions~\ref{ass:compactness} and~\ref{ass:volatility_process}; population identification from Assumption~\ref{ass:identification_instrumental}; and finite-history forgetting from Assumption~\ref{ass:approximation}. The smoothness and necessary moment conditions hold by Assumptions~\ref{ass:moments} with $r=2$ and~\ref{ass:differentiability_instrumental}, whereas Assumptions~\ref{ass:moments} with $r=4$ and~\ref{ass:qml_moments} meet the additional fourth-order requirements. We next establish the asymptotic linear representation. Almost surely, the interior QML estimator satisfies
\[
 0=\frac{1}{\sqrt{n}}\sum_{t=1}^n\partial_{\bs\theta} g\big(y_t,\widetilde\sigma_t(\bs\theta_0)\big)+\left[\frac{1}{n}\sum_{t=1}^n \partial_{\bs\theta\bs\theta'}^2 g\big(y_t,\widetilde\sigma_t(\bs\theta_n^\ast)\big)\right]\sqrt{n}(\widehat{\bs\theta}_n-\bs\theta_0),
\]
where $\bs\theta_n^\ast$ lies between $\widehat{\bs\theta}_n$ and $\bs\theta_0$. It then follows from Lemma~\ref{lem:approx} that
\[
 \frac{1}{\sqrt{n}}\sum_{t=1}^n\partial_{\bs\theta}g\big(y_t,\widetilde\sigma_t(\bs\theta_0)\big)=\frac{1}{\sqrt{n}}\sum_{t=1}^n\bs{s}_t+o_p(1).
\]
By the ergodic theorem, the conditions for QML consistency, the local Hessian envelope and the finite-history Hessian equivalence ensure that
\[
 -\frac{1}{n}\sum_{t=1}^n \partial^2_{\bs\theta\bs\theta'} g\{y_t,\widetilde\sigma_t(\bs\theta_n^\ast)\}=\bs{H}_0+o_p(1).
\]
Because $\bs{H}_0=\kappa_h\bs{I}_0$ is positive definite and $n^{-1/2}\sum_{t=1}^n \bs{s}_t=O_p(1)$ by the martingale central limit theorem, inversion of the preceding first-order expansion yields \eqref{eq:qMLE_linear_rep}.
\end{proof}

\subsection{Consistency of the unconditional expectile}

\begin{proof}[Proof of Theorem~\ref{thm:consistency_expectiles}]
Let $\widehat\xi_n=\xpb{\widehat{\bs\theta}_n}$ be the unique zero of $\widetilde Q_n(\cdot,\widehat{\bs\theta}_n)$.  For fixed $\xi$, it follows from the global Lipschitz property in Lemma~\ref{lem:expectile_identity}, the mean-value theorem, and the ergodic theorem that $|Q_n(\xi,\widehat{\bs\theta}_n)-Q_n(\xi,\bs\theta_0)|\le C\|\widehat{\bs\theta}_n-\bs\theta_0\|\frac{1}{n}\sum_{t=1}^n|\eta_t|B_t\asconv 0$. By Lemma~\ref{lem:approx}, the difference between $Q_n$ and $\widetilde Q_n$ is almost surely negligible at the limit. Moreover,
\[
 Q_n(\xi,\bs\theta_0)=\frac{1}{n}\sum_{t=1}^n\phi_\tau(\eta_t-\xi) \asconv Q(\xi):=\E\phi_\tau(\eta_t-\xi)
\]
for each $\xi$. The map $Q$ is continuous and strictly decreasing, and its unique zero is $\xi_0$. See \cn{holzmann2016expectile} and \cn{kratschmer2017statistical}. In turn, given that $Q(\xi_0-\varepsilon)>0>Q(\xi_0+\varepsilon)$ for any $\varepsilon>0$, the preceding pointwise convergence implies that $\widetilde Q_n(\xi_0-\varepsilon,\widehat{\bs\theta}_n)>0> \widetilde Q_n(\xi_0+\varepsilon,\widehat{\bs\theta}_n)$ almost surely for sufficiently large $n$. Since $\widetilde Q_n(\cdot,\widehat{\bs\theta}_n)$ is continuous and strictly decreasing, its zero lies between these two points, thereby establishing that $\widehat\xi_n\asconv\xi_0$.
\end{proof}

\subsection{Bahadur representation and joint central limit theorem}

\begin{lemma}[Bahadur representation]\label{lem:bahadur}
Under the conditions in Theorem~\ref{thm:asymptotic_expectiles},
\begin{equation}\label{eq:bahadur_correct}
 \sqrt{n}(\widehat\xi_n-\xi_0) =\Psi_0^{-1}\frac{1}{\sqrt{n}}\sum_{t=1}^n\psi_t-\xi_0\bs{J}_0'\sqrt{n}(\widehat{\bs\theta}_n-\bs\theta_0)+o_p(1).
\end{equation}
\end{lemma}

\begin{proof}
We first establish the rate of $\widehat\xi_n$ without using the Bahadur representation.  Proposition~\ref{prop:consistency_qmle} dictates that $\sqrt{n}(\widehat{\bs\theta}_n-\bs\theta_0)=O_p(1)$.  Consider $M_\theta<\infty$ such that $\Pr\!\left\{\sqrt{n}\|\widehat{\bs\theta}_n-\bs\theta_0\|>M_\theta\right\}<\varepsilon$ for a fixed $\varepsilon>0$ and sufficiently large $n$. In the complementary event, Lemma~\ref{lem:local_expectile_expansion} applies uniformly to $(\widehat{\bs\theta}_n,\xi_0+v/\sqrt{n})$ for every fixed $|v|\le M_\xi$, where $M_\xi$ is arbitrary. It then follows from Lemma~\ref{lem:approx} and $Q_n(\xi_0,\bs\theta_0)=n^{-1}\sum_t\psi_t$ that
\begin{align}
 \sqrt{n}\,\widetilde Q_n \left(\xi_0+\frac{v}{\sqrt{n}},\widehat{\bs\theta}_n\right)
 &= \frac{1}{\sqrt{n}}\sum_{t=1}^n\psi_t -\xi_0\Psi_0\bs{J}_0'\sqrt{n}(\widehat{\bs\theta}_n-\bs\theta_0) -v\Psi_0+o_p(1),                         \label{eq:bahadur_endpoint}
\end{align}
uniformly for $|v|\le M_\xi$ in such an event. The first term in \eqref{eq:bahadur_endpoint} is tight by the iid central limit theorem, whereas the QMLE representation ensures the tightness of second term. Because $\Psi_0>0$, we can choose $M_\xi$ so large that, with probability at least $1-2\varepsilon$, $\widetilde Q_n\left(\xi_0-\frac{M_\xi}{\sqrt{n}},~\widehat{\bs\theta}_n\right)>0$ and $\widetilde Q_n\left(\xi_0+\frac{M_\xi}{\sqrt{n}},~\widehat{\bs\theta}_n\right)<0$ for all large $n$. The map $\xi\mapsto\widetilde Q_n(\xi,\widehat{\bs\theta}_n)$ is continuous and strictly decreasing, with unique zero between the two endpoints. As $\varepsilon$ is arbitrary, $\sqrt{n}(\widehat\xi_n-\xi_0)=O_p(1)$. Since $\sqrt{n}(\widehat{\bs\theta}_n-\bs\theta_0)=O_p(1)$ and $\sqrt{n}(\widehat\xi_n-\xi_0)=O_p(1)$, for every $\varepsilon>0$, one can choose a fixed $M<\infty$ such that $\limsup_{n\to\infty}\Pr\{(\widehat{\bs\theta}_n,\widehat\xi_n)\notin\mathcal{N}_n(M)\}<\varepsilon$. For every $\delta>0$, it then follows that
\[
 \Pr\{|R_n(\widehat{\bs\theta}_n,\widehat\xi_n)|>\delta\} \le \Pr\{(\widehat{\bs\theta}_n,\widehat\xi_n) \notin\mathcal{N}_n(M)\} + \Pr\left\{ \sup_{(\bs\theta,\xi)\in\mathcal{N}_n(M)} |R_n(\bs\theta,\xi)|>\delta \right\},
\]
where $R_n(\bs\theta,\xi)$ denotes the remainder in the uniform local expansion. The second probability converges to zero by the uniform expansion, whereas we can make the first arbitrarily small by increasing $M$. As such, the remainder is $o_p(1)$ at any random pair. Because $\widetilde Q_n(\widehat\xi_n,\widehat{\bs\theta}_n)=0$, Lemma~\ref{lem:approx} gives $0=\frac{1}{\sqrt{n}}\sum_{t=1}^n\psi_t-\xi_0\Psi_0\bs{J}_0'\sqrt{n}(\widehat{\bs\theta}_n-\bs\theta_0)-\Psi_0\sqrt{n}(\widehat\xi_n-\xi_0)+o_p(1)$. Dividing by the strictly positive constant $\Psi_0$ obtains \eqref{eq:bahadur_correct}. 
\end{proof}

\begin{lemma}[Joint central limit theorem]\label{lem:joint_clt}
Let $\bs\Sigma_s=\E[\bs{s}_t\bs{s}_t']$, $\bs\sigma_{s\psi}=\E[\bs{s}_t\psi_t]$ and $\sigma_\psi^2=\E[\psi_t^2]$. Under the conditions in Proposition~\ref{prop:consistency_qmle}, it follows that
\begin{equation}\label{eq:joint_clt_correct}
 \begin{pmatrix}
    n^{-1/2}\sum_{t=1}^n\psi_t\\
    \sqrt{n}(\widehat{\bs\theta}_n-\bs\theta_0)
 \end{pmatrix}
 \dconv\mathcal{N}\left(0,\begin{pmatrix}
   \sigma_\psi^2          & \bs\sigma_{\psi\theta}'\\
   \bs\sigma_{\psi\theta} & \bs\Sigma_\theta
 \end{pmatrix}\right),
\end{equation}
where $\bs\Sigma_\theta=\bs{H}_0^{-1}\Sigma_s\bs{H}_0^{-1}$ and $\bs\sigma_{\psi\theta}=\bs{H}_0^{-1}\bs\sigma_{s\psi}$. Moreover, $\bs\Sigma_s=\E[a(\eta_t)^2]\bs{I}_0$ and $\bs\sigma_{s\psi}=\E[a(\eta_t)\psi_t]\bs{J}_0$ given that $\eta_t$ is independent of $\mathcal{F}_{t-1}$.
\end{lemma}

\begin{proof}
Let $\bs{Z}_t=(\psi_t,\bs{s}_t')'\in\mathbb{R}^{m+1}$, which is $\mathcal{F}_t$-measurable. In addition, $\E(\bs{Z}_t\mid\mathcal{F}_{t-1})=0$ because $\bs{D}_t$ is $\mathcal{F}_{t-1}$-measurable, $\eta_t$ is independent of $\mathcal{F}_{t-1}$, $\E[\psi_t]=0$ and $\E[a(\eta_t)]=0$. As such, $\{\bs{Z}_t,\mathcal{F}_t\}$ is a stationary martingale-difference sequence. We next verify the martingale central limit theorem through the Cramér-Wold device. Fix $\bs\lambda\in\mathbb{R}^{m+1}$ and set $q_t(\bs\lambda)=\bs\lambda'\bs{Z}_t$. The conditional variance of $q_t(\bs\lambda)$ then is $\E[q_t(\bs\lambda)^2\mid\mathcal{F}_{t-1}]=\bs\lambda'\E(\bs{Z}_t\bs{Z}_t'\mid\mathcal{F}_{t-1})\bs\lambda$, where
\[
 \E(\bs{Z}_t\bs{Z}_t'\mid\mathcal{F}_{t-1}) 
 =\begin{pmatrix}
 \E[\psi_t^2]                && \E[\psi_ta(\eta_t)]\bs{D}_t'\\
 \E[\psi_ta(\eta_t)]\bs{D}_t && \E[a(\eta_t)^2]\bs{D}_t\bs{D}_t'
 \end{pmatrix}.
\]
All entries are stationary and ergodic functions of the past, as well as integrable by Assumption~\ref{ass:moments} with $r=4$ and Assumption~\ref{ass:qml_moments}. By the ergodic theorem, $\frac{1}{n}\sum_{t=1}^n\E\{q_t(\bs\lambda)^2\mid\mathcal{F}_{t-1}\}\asconv\bs\lambda'\bs\Sigma_Z\bs\lambda$, where
 $\bs\Sigma_Z=
 \begin{bsmallmatrix}
     \sigma_\psi^2     & \bs\sigma_{s\psi}'\\
     \bs\sigma_{s\psi} & \bs\Sigma_s
 \end{bsmallmatrix}$. For every $\varepsilon>0$, the conditional Lindeberg condition
\[
 \frac{1}{n}\sum_{t=1}^n\E\big[q_t^2(\bs\lambda)\bs{1}\{|q_t(\bs\lambda)|>\varepsilon\sqrt{n}\}\mid\mathcal{F}_{t-1}\big]\le\frac{1}{\varepsilon^2 n}\, \frac{1}{n}\sum_{t=1}^n\E\big[q_t^4(\bs\lambda)\mid\mathcal{F}_{t-1}\big]
\]
converges to zero in probability given the average in the right-hand side is $O_p(1)$ by the ergodic theorem. In fact, $\E\|\bs{Z}_t\|^4<\infty$ follows from $|\psi_t|\le\bar\tau(|\eta_t|+|\xi_0|)$, $\|\bs{s}_t\|=|a(\eta_t)|\|\bs{D}_t\|$, independence between $\eta_t$ and $\bs{D}_t$, Assumption~\ref{ass:moments} with $r=4$, and Assumption~\ref{ass:qml_moments}. The martingale central limit theorem \ca{hallheyde1980}{see, for example,} yields
$\bs\lambda'\frac{1}{\sqrt{n}}\sum_{t=1}^n\bs{Z}_t\dconv\mathcal{N}(0,\bs\lambda'\bs\Sigma_Z\bs\lambda)$ for every $\bs\lambda$. The Cramér-Wold device \ca{vandervaart1998asymptotic}{see, e.g.,} then implies $\frac{1}{\sqrt{n}}\sum_{t=1}^n\bs{Z}_t\dconv\mathcal{N}(0,\bs\Sigma_Z)$. Combining this result with \eqref{eq:qMLE_linear_rep} and Slutsky's theorem proves \eqref{eq:joint_clt_correct}. Finally, independence between $\eta_t$ and $\bs{D}_t$ implies  $\bs\Sigma_s=\E[a(\eta_t)^2]\bs{I}_0$ and $\bs\sigma_{s\psi}=\E[a(\eta_t)\psi_t]\bs{J}_0$. Finally, we can use the fact that $\bs{H}_0=\kappa_h\bs{I}_0$ to obtain $\bs\Sigma_\theta=\frac{\E[a(\eta_t)^2]}{\kappa_h^2}\,\bs{I}_0^{-1}$ and $\bs\sigma_{\psi\theta}=\frac{\E[a(\eta_t)\psi_t]}{\kappa_h}\,\bs{I}_0^{-1}\bs{J}_0$.
\end{proof}

\subsection{Asymptotic distribution of the unconditional expectile}

\begin{proof}[Proof of Theorem~\ref{thm:asymptotic_expectiles}]
It follows from Lemmata~\ref{lem:bahadur} and~\ref{lem:joint_clt} that the influence representation reads
\[
 \sqrt{n}(\widehat\xi_n-\xi_0) =\frac{1}{\sqrt{n}}\sum_{t=1}^n \left\{\Psi_0^{-1}\psi_t -\xi_0\bs{J}_0' \bs{H}_0^{-1}\bs{s}_t\right\}+o_p(1)
\]
and hence $\sqrt{n}(\widehat\xi_n-\xi_0)\dconv\mathcal{N}(0,V_\xi)$, where $V_\xi=\Psi_0^{-2}\sigma_\psi^2+\xi_0^2\bs{J}_0'\Sigma_\theta\bs{J}_0-2\xi_0\Psi_0^{-1}\bs{J}_0'{\bs\sigma}_{\psi\theta}$ given that
\[
 \operatorname{Cov}\Big(n^{-1/2}\sum_{t=1}^n\psi_t,\,\sqrt{n}(\widehat{\bs\theta}_n-\bs\theta_0)\Big)\longrightarrow\bs{H}_0^{-1}\E(\bs{s}_t\psi_t)=\bs{\bs\sigma}_{\psi\theta}.
\]
Plugging in the expressions for $\bs\Sigma_\theta$ and $\bs\sigma_{\psi\theta}$ then yields a more compact formula for the asymptotic variance: namely,
\[
 V_\xi =\frac{\sigma_\psi^2}{\Psi_0^2} +\left\{ \frac{\xi_0^2\E[a(\eta_t)^2]}{\kappa_h^2} -\frac{2\xi_0\E[a(\eta_t)\psi_t]}{\Psi_0\kappa_h} \right\}\bs{J}_0'\bs{I}_0^{-1}\bs{J}_0.
\]
\end{proof}

\subsection{Conditional expectile at the observed forecast state}

Due to the positive homogeneity of the expectile and independence between $\eta_{n+1}$ and $\mathcal{F}_n$, the one-step-ahead conditional expectile reads $c_{n+1}=\sigma_{n+1}(\bs\theta_0)\xi_0$, which we estimate by $\widehat{c}_{n+1}=\widetilde\sigma_{n+1}(\widehat{\bs\theta}_n)\widehat\xi_n$. For $k\ge1$, let $\mathrm{BL}_1(\mathbb{R}^k)=\left\{f:\mathbb{R}^k\to\mathbb{R}~\mbox{such that}~\|f\|_\infty\le1,~|f(u)-f(v)|\le\|u-v\|\right\}$, and define
\[
 d_{BL}^{(k)}(P,Q)=\sup_{f\in\mathrm{BL}_1(\mathbb{R}^k)}\left|\int f\,\mbox{d}P-\int f\,\mbox{d}Q\right|,
\]
with $d_{BL}:=d_{BL}^{(1)}$. We employ $d_{BL}^{(k)}$ because it metrizes weak convergence on Euclidean spaces \cite{vandervaartwellner1996weak}. Define $\bs{Y}_n:=(n^{-1/2}\sum_{t=1}^n\psi_t,\,\sqrt{n}(\widehat{\bs\theta}_n-\bs\theta_0))'$ and denote by $\bs\Sigma_Y$ its covariance matrix in \eqref{eq:joint_clt_correct}. Let $\bs{D}_{n+1}=\partial_{\bs\theta}\log\sigma_{n+1}(\bs\theta_0)$ and $V_{n+1}=\bs{A}_{n+1}\bs\Sigma_Y\bs{A}_{n+1}'$, with $\bs{A}_{n+1}=\sigma_{n+1}(\bs\theta_0)\left(\Psi_0^{-1},\,\xi_0(\bs{D}_{n+1}-\bs{J}_0)'\right)$. For the finite-memory state in Assumption~\ref{ass:forecast}, define analogously $V_{n+1}^{[m]}=\bs{A}_{n+1}^{[m]}\bs\Sigma_Y\bs{A}_{n+1}^{[m]'}$ and $\bs{A}_{n+1}^{[m]}=\sigma_{n+1}^{[m]}\left(\Psi_0^{-1},\,\xi_0(\bs{D}_{n+1}^{[m]}-\bs{J}_0)'\right)$. Finally, let $m_n\to\infty$ with $m_n/n\to0$ and write $\mathcal{H}_n=\mathcal{H}_{n,m_n}$.

\begin{proof}[Proof of Theorem~\ref{thm:conditional_expectiles}]
Letting $\bs\Delta_{\theta,n}=\widehat{\bs\theta}_n-\bs\theta_0$ and $\Delta_{\xi,n}=\widehat\xi_n-\xi_0$ yields by Taylor's theorem
\[
 \sigma_{n+1}(\widehat{\bs\theta}_n)=\sigma_{n+1}(\bs\theta_0)+\sigma_{n+1}(\bs\theta_0)\bs{D}_{n+1}'\bs\Delta_{\theta,n}+\mathcal R_{\sigma,n},
\]
where $|\mathcal R_{\sigma,n}|\le\frac{1}{2}\sup_{\bs\theta\in V}\left\|\partial^2_{\bs\theta\bs\theta'}\sigma_{n+1}(\bs\theta)\right\|\|\bs\Delta_{\theta,n}\|^2$ in the event $\widehat{\bs\theta}_n\in V$. Since $\widehat\xi_n=\xi_0+\Delta_{\xi,n}$,
$$
\sigma_{n+1}(\widehat{\bs\theta}_n)\widehat\xi_n-\sigma_{n+1}(\bs\theta_0)\xi_0=\sigma_{n+1}(\bs\theta_0) \left[\Delta_{\xi,n}+\xi_0\bs{D}_{n+1}'\bs\Delta_{\theta,n}+\bs{D}_{n+1}'\bs\Delta_{\theta,n}\bs\Delta_{\xi,n}\right]+\mathcal R_{\sigma,n}(\xi_0+\Delta_{\xi,n}).
$$
The remainder term is $o_p(n^{-1/2})$ given that $\sqrt{n}\,|\mathcal R_{\sigma,n}| (|\xi_0|+|\Delta_{\xi,n}|)=o_p(1)$ by Assumption~\ref{ass:forecast}. In addition, the last term within brackets is also of order $o_p(n^{-1/2})$ in view that
\[
 \sqrt{n}\,\sigma_{n+1}(\bs\theta_0)\|\bs{D}_{n+1}\| \|\bs \Delta_{\theta,n}\||\Delta_{\xi,n}|=\sigma_{n+1}(\bs\theta_0)\|\bs{D}_{n+1}\| \frac{ \|\sqrt{n}\,\bs\Delta_{\theta,n}\| |\sqrt{n}\,\Delta_{\xi,n}|}{\sqrt{n}}=o_p(1).
\]
Consequently, 
\begin{equation}\label{eq:observed_forecast_delta}
 \sqrt{n}\left[\sigma_{n+1}(\widehat{\bs\theta}_n)\widehat\xi_n-\sigma_{n+1}(\bs\theta_0)\xi_0\right]=\sigma_{n+1}(\bs\theta_0)\sqrt{n}\,\Delta_{\xi,n}+\xi_0\sigma_{n+1}(\bs\theta_0)\bs{D}_{n+1}'\sqrt{n}\,\bs\Delta_{\theta,n}+o_p(1).
\end{equation}
It also follows from Assumption~\ref{ass:approximation} that $\sqrt{n}\left|\widetilde\sigma_{n+1}(\widehat{\bs\theta}_n)-\sigma_{n+1}(\widehat{\bs\theta}_n) \right||\widehat\xi_n|=o_p(1)$. The Bahadur representation then gives $\sqrt{n}\,\Delta_{\xi,n}=\Psi_0^{-1}\frac{1}{\sqrt{n}}\sum_{t=1}^n\psi_t-\xi_0\bs{J}_0'\sqrt{n}\,\bs\Delta_{\theta,n}+o_p(1)$. Plugging this expression into \eqref{eq:observed_forecast_delta} yields
\begin{align*}
 \sqrt{n}(\widehat c_{n+1}-c_{n+1})&=\sigma_{n+1}(\bs\theta_0)\Psi_0^{-1}\frac{1}{\sqrt{n}}\sum_{t=1}^n\psi_t+\xi_0\sigma_{n+1}(\bs\theta_0)(\bs{D}_{n+1}-\bs{J}_0)'\sqrt{n}\,\bs\Delta_{\theta,n}+o_p(1)\\
 &=\bs{A}_{n+1}\bs{Y}_n+o_p(1).
\end{align*}
To handle the dependence between $\bs{A}_{n+1}$ and $\bs{Y}_n$, define $\bs{Y}_t^{\star}=\big(\psi_t,\,\bs{H}_0^{-1}\bs{s}_t\big)'$, $\bs{Y}_{n,0}=\frac{1}{\sqrt{n}}\sum_{t=1}^{n-m_n}\bs{Y}_t^{\star}$, and $\bs{Y}_{n,1}=\frac{1}{\sqrt{n}}\sum_{t=n-m_n+1}^{n}\bs{Y}_t^{\star}$. By \eqref{eq:qMLE_linear_rep}, $\bs{Y}_n=\bs{Y}_{n,0}+\bs{Y}_{n,1}+o_p(1)$. In addition, since $\{\bs{Y}_t^{\star},\mathcal{F}_t\}$ is a square-integrable martingale-difference sequence, $\E\|\bs{Y}_{n,1}\|^2=\frac{m_n}{n}\E\|\bs{Y}_0^{\star}\|^2\longrightarrow0$, so that $\bs{Y}_{n,1}=o_p(1)$ and $\bs{Y}_{n,0}\dconv\bs{Y}\sim\mathcal{N}(0,\bs\Sigma_Y)$. By causality, $\bs{Y}_{n,0}$ is measurable with respect to $\sigma(\eta_t:t\le n-m_n)$, while $\bs{A}_{n+1}^{[m_n]}$ is $\mathcal{H}_n$-measurable. These two random objects are independent. Moreover, \eqref{eq:forecast_terminal_contraction} and tightness of the terminal coefficients imply
$\|\bs{A}_{n+1}^{[m_n]}-\bs{A}_{n+1}\|=o_p(1)$. It then follows that
\begin{equation}\label{eq:observed_forecast_block}
 \sqrt{n}(\widehat c_{n+1}-c_{n+1})=\bs{A}_{n+1}^{[m_n]}\bs{Y}_{n,0}+r_n,
 \qquad\mbox{with}~r_n=o_p(1).
\end{equation}
For any deterministic row vector $\bs{a}$,
\begin{equation}\label{eq:observed_BL_linear_map}
 d_{BL}\Big(\mathcal{L}(\bs{a}\bs{Y}_{n,0}),\mathcal{L}(\bs{a}\bs{Y})\Big)\le\max(1,\|\bs{a}\|)d_{BL}^{(m+1)}\Big(\mathcal{L}(\bs{Y}_{n,0}),\mathcal{L}(\bs{Y})\Big).
\end{equation}
Conditional on $\mathcal{H}_n$, $\bs{A}_{n+1}^{[m_n]}$ is fixed and $\bs{Y}_{n,0}$ has its unconditional law. Since $\|\bs{A}_{n+1}^{[m_n]}\|=O_p(1)$, \eqref{eq:observed_BL_linear_map} and a truncation argument give $d_{BL}\left(\mathcal{L}(\bs{A}_{n+1}^{[m_n]}\bs{Y}_{n,0}\mid\mathcal{H}_n),\mathcal{N}(0,V_{n+1}^{[m_n]})\right)\pconv0$. Finally,
$$
d_{BL}\Big(\mathcal{L}(X+r_n\mid\mathcal{H}_n),\mathcal{L}(X\mid\mathcal{H}_n)\Big)\le\E\big[\min(2,|r_n|)\mid\mathcal{H}_n\big],
$$
which converges to zero in probability because $r_n=o_p(1)$. Bearing \eqref{eq:observed_forecast_block} in mind, this proves that $d_{BL}\Big(\mathcal{L}\big(\sqrt{n}(\widehat{c}_{n+1}-c_{n+1})\mid\mathcal{H}_n\big),\mathcal{N}(0,V_{n+1}^{[m_n]})\Big)\pconv 0$. To wrap up the proof, it suffices to appreciate that $V_{n+1}^{[m_n]}-V_{n+1}\pconv0$ by \eqref{eq:forecast_terminal_contraction} and the continuity of $\bs{A}\mapsto \bs{A}\bs\Sigma_Y\bs{A}'$.
\end{proof}

\subsection{Consistent estimation of the asymptotic variance}

To establish the plug-in estimators for each component of the asymptotic variance, we first define $\widehat{\bs{D}}_t=\partial_{\bs\theta}\log\widetilde\sigma_t(\widehat{\bs\theta}_n)$ and then let $\widehat\eta_t=y_t/\widetilde\sigma_t(\widehat{\bs\theta}_n)$, $\widehat\psi_t=\phi_\tau(\widehat\eta_t-\widehat\xi_n)$, $\widehat w_t=w_\tau(\widehat\eta_t-\widehat\xi_n)$, $\widehat{\bs{s}}_t=a(\widehat\eta_t)\widehat{\bs{D}}_t$, $\widehat\Psi_n=\frac{1}{n}\sum_{t=1}^n\widehat w_t$, $\widehat{\bs{J}}_n=\frac{1}{n}\sum_{t=1}^n\widehat {\bs{D}}_t$, $\widehat{\bs{I}}_n=\frac{1}{n}\sum_{t=1}^n\widehat {\bs{D}}_t\widehat {\bs{D}}_t'$, $\widehat\kappa_h=-\frac{1}{n}\sum_{t=1}^n b(\widehat\eta_t)$, and $\widehat{\bs{H}}_{F,n}=\widehat\kappa_h\,\widehat {\bs{I}}_n$. The latter includes a subscript $F$ to emphasize the use of the factorization $\bs{H}_0=\kappa_h\bs{I}_0$ rather than a sample average of the full QML Hessian. We are now ready to define $\widehat{\bs\Sigma}_s =\frac{1}{n}\sum_{t=1}^n\widehat{\bs{s}}_t\widehat{\bs{s}}_t'$, $\widehat{\bs\sigma}_{s\psi}=\frac{1}{n}\sum_{t=1}^n\widehat{\bs{s}}_t\widehat\psi_t$,  $\widehat\sigma_\psi^2=\frac{1}{n}\sum_{t=1}^n\widehat\psi_t^2$, $\widehat{\bs\Sigma}_\theta=\widehat{\bs{H}}_{F,n}^{-1}\widehat{\bs\Sigma}_s \widehat{\bs{H}}_{F,n}^{-1}$, and $\widehat{\bs\sigma}_{\psi\theta}=\widehat{\bs{H}}_{F,n}^{-1}\widehat{\bs\sigma}_{s\psi}$.

\begin{proposition}[Consistent variance estimation]
\label{prop:variance_estimation}
Under the conditions in Theorem~\ref{thm:asymptotic_expectiles} and Assumption~\ref{ass:plugin},
\begin{equation}\label{eq:Vhat_component}
 \widehat{V}_\xi=\widehat\Psi_n^{-2}\widehat\sigma_\psi^2+\widehat\xi_n^2\widehat{\bs{J}}_n'\widehat{\bs\Sigma}_\theta\widehat{\bs{J}}_n -2\widehat\xi_n\widehat\Psi_n^{-1}\widehat{\bs{J}}_n'\widehat{\bs\sigma}_{\psi\theta}\pconv{V}_\xi.
\end{equation}
Alternatively, if we define $\widehat\varphi_t=\widehat\Psi_n^{-1}\widehat\psi_t-\widehat\xi_n\widehat{\bs{J}}_n'(\widehat{\bs{H}}_{F,n})^{-1}\widehat{\bs{s}}_t$, then $\widehat V_{F,\xi}=\frac{1}{n}\sum_{t=1}^n(\widehat\varphi_t-\overline{\widehat\varphi})^2\pconv V_\xi$.
\end{proposition}

\begin{proof}
We start with the consistency of the primitive quantities and then apply the continuous mapping theorem. We do so in five steps. The first concerns the removal of the finite-history initialization, whereas the second handles the derivatives and factorized Hessian. The third and fourth steps deal with the expectile- and QML-score moments. Finally, the last step applies the continuous mapping theorem.\vskip 1em

\noindent\emph{Step 1}\\
In the event that $(\widehat{\bs\theta}_n,\widehat\xi_n)\in V\times K$, for a fixed compact $K$ containing $\xi_0$, Lemma~\ref{lem:approx} applies uniformly to the smooth QML quantities and to $\widehat\psi_t$. The probability that $(\widehat{\bs\theta}_n,\widehat\xi_n)\in V\times K$ converges to one by Proposition~\ref{prop:consistency_qmle} and Theorem~\ref{thm:consistency_expectiles}. Accordingly, it suffices to establish the limits in the enunciate with $\eta_t^\circ=A_t(\widehat{\bs\theta}_n)\eta_t$ and $\bs{D}_t^\circ=\bs{D}_t(\widehat{\bs\theta}_n)$ in lieu of the finite-history residuals and derivatives.
However, the indicator weight requires a separate argument because it is not Lipschitz. It turns out that, for any real $u$ and $v$, $|w_\tau(u)-w_\tau(v)|\le c_\tau\bs{1}\{|u|\le|u-v|\}$ and hence
\begin{equation}
 |\widehat w_t-w_t|\le c_\tau\bs{1}\big\{|r_t|\le|\widehat\eta_t-\eta_t|+|\widehat\xi_n-\xi_0|\big\},
 \label{eq:w_indicator_bound}
\end{equation}
with $r_t=\eta_t-\xi_0$. Assume now that $\|\widehat{\bs\theta}_n-\bs\theta_0\|+|\widehat\xi_n-\xi_0|\le\epsilon$ and let $a_{t,\epsilon}=C\epsilon(1+B_t|\eta_t|)$ and $e_t=C_0|y_t|\rho^t/\underline\sigma^2$. Because $\bs{1}\{|r_t|\le a_{t,\epsilon}+e_t\}\le\bs{1}\{|r_t|\le 2a_{t,\epsilon}\}+\bs{1}\{|r_t|\le2e_t\}$, we can treat stationary and initialization effects separately. The expectation of the first indicator converges to zero as $\epsilon\downarrow0$ by dominated convergence given that $\Pr(\eta_t=\xi_0)=0$. This means the ergodic theorem controls its sample average. For the initialization term, $\Pr(|r_t|\le e_t)\le\Pr(|r_t|\le a)+a^{-1}\E[e_t]$ for every $a>0$. Now, we first let $t\to\infty$, bearing in mind that $\E[e_t]=O(\rho^t)$, and then let $a\downarrow0$. This ensures that $\Pr(|r_t|\le e_t)\to0$, while Cesàro averaging yields $\frac{1}{n}\sum_{t=1}^n\Pr(|r_t|\le e_t)\to0$. It then follows from Markov's inequality and \eqref{eq:w_indicator_bound} that 
\begin{equation}\label{eq:Psi_consistency_detail}
 \frac{1}{n}\sum_{t=1}^n|\widehat w_t-w_t|\pconv0\quad\mbox{and}\quad\widehat\Psi_n\pconv\Psi_0.
\end{equation}\vskip 1em

\noindent\emph{Step 2}\\
By the derivative envelope in Assumption~\ref{ass:moments}, $\frac{1}{n}\sum_{t=1}^n\|{\bs{D}}_t^\circ-{\bs{D}}_t\| \le \|\widehat{\bs\theta}_n-\bs\theta_0\| \frac{1}{n}\sum_{t=1}^n B_t=o_p(1)$. The same argument in squared norm, with $r=4$ in Assumption~\ref{ass:moments}, gives way to $\widehat{\bs{J}}_n\pconv\bs{J}_0$ and $\widehat{\bs{I}}_n\pconv\bs{I}_0$. By Assumption~\ref{ass:plugin}, $\frac{1}{n}\sum_{t=1}^n|b(\eta_t^\circ)-b(\eta_t)| \le\|\widehat{\bs\theta}_n-\bs\theta_0\|\frac{1}{n}\sum_{t=1}^n L_{b,t}=o_p(1)$. As such, the ergodic theorem and $\E|b(\eta_t)|<\infty$ imply
\begin{equation}
 \widehat\kappa_h\pconv\kappa_h=-\E[b(\eta_t)],\quad\mbox{and}\quad\widehat{\bs{H}}_{F,n}\pconv\kappa_h\bs{I}_0=\bs{H}_0.
 \label{eq:Hfactor_consistency}
\end{equation}
Since $\kappa_h>0$ and $\bs{I}_0$ is positive definite, $\widehat{\bs{H}}_{F,n}$ is invertible with probability approaching one and $\widehat{\bs{H}}_{F,n}^{-1}\pconv\bs{H}_0^{-1}$.\vskip 1em

\noindent\emph{Step 3}\\
Lemma~\ref{lem:expectile_identity} gives $|\psi_{t,\tau}(\widehat{\bs\theta}_n,\widehat\xi_n)-\psi_t|\le\bar\tau\left(|\eta_t^\circ-\eta_t|+|\widehat\xi_n-\xi_0|\right)$. Furthermore,
\[
 \frac{1}{n}\sum_{t=1}^n|\eta_t^\circ-\eta_t|^2 \le C\|\widehat{\bs\theta}_n-\bs\theta_0\|^2 \frac{1}{n}\sum_{t=1}^n B_t^2\eta_t^2=o_p(1),
\]
where the last average is of order $O_p(1)$ given the moment conditions and the ergodic theorem. This means that
\begin{equation}
 \frac{1}{n}\sum_{t=1}^n|\widehat\psi_t-\psi_t|^2\pconv0.
 \label{eq:psi_L2_plugin}
\end{equation}
By Cauchy--Schwarz,
\begin{align*}
 \left|\frac{1}{n}\sum_{t=1}^n (\widehat\psi_t^2-\psi_t^2)\right| &\le \left\{\frac{1}{n}\sum_{t=1}^n |\widehat\psi_t-\psi_t|^2\right\}^{1/2} \left\{\frac{1}{n}\sum_{t=1}^n (|\widehat\psi_t|+|\psi_t|)^2\right\}^{1/2},
\end{align*}
which is $o_p(1)$ given that the second factor is $O_p(1)$ by Assumption~\ref{ass:differentiability_instrumental} and the moment conditions on the innovations. Altogether, this ensures that $\widehat\sigma_\psi^2\pconv\E[\psi_t]^2=\sigma_\psi^2$.\vskip 1em

\noindent\emph{Step 4}\\
From Assumption~\ref{ass:plugin} and the derivative bounds,
\begin{align*}
 \frac{1}{n}\sum_{t=1}^n\|\widehat s_t-s_t\|^2 &\le C\frac{1}{n}\sum_{t=1}^n \left[ |a(\eta_t^\circ)-a(\eta_t)|^2\|D_t^\circ\|^2+a(\eta_t)^2\|{\bs{D}}_t^\circ-{\bs{D}}_t\|^2 \right]+o_p(1) =o_p(1).
\end{align*}
The last equality follows from the consistency of $\widehat{\bs\theta}_n$, the Lipschitz envelopes, Assumption~\ref{ass:moments} with $r=4$, and Assumption~\ref{ass:differentiability_instrumental}. Applying then Cauchy-Schwarz gives $\widehat\Omega_s=\frac{1}{n}\sum_{t=1}^n\widehat{\bs{s}}_t \widehat{\bs{s}}_t'\pconv\E({\bs{s}}_t{\bs{s}}_t')=\bs\Sigma_s$. Similarly, using \eqref{eq:psi_L2_plugin},
\begin{align*}
 \left\|\frac{1}{n}\sum_{t=1}^n \widehat {\bs{s}}_t\widehat\psi_t -\frac{1}{n}\sum_{t=1}^n{\bs{s}}_t\psi_t\right\|
 &\le \left\{\frac{1}{n}\sum_{t=1}^n \|\widehat {\bs{s}}_t-{\bs{s}}_t\|^2\right\}^{1/2} \left\{\frac{1}{n}\sum_{t=1}^n\widehat\psi_t^2\right\}^{1/2}\\
 &\quad+ \left\{\frac{1}{n}\sum_{t=1}^n\|{\bs{s}}_t\|^2\right\}^{1/2} \left\{\frac{1}{n}\sum_{t=1}^n |\widehat\psi_t-\psi_t|^2\right\}^{1/2} =o_p(1).
\end{align*}
The ergodic theorem therefore yields
\begin{equation}
 \widehat {\bs\sigma}_{s\psi}\pconv\E({\bs{s}}_t\psi_t)={\bs\sigma}_{s\psi}.
 \label{eq:cross_consistency_detail}   
\end{equation}\vskip 1em
 
\noindent\emph{Step 5}\\
Together, \eqref{eq:Psi_consistency_detail}, \eqref{eq:Hfactor_consistency} and \eqref{eq:cross_consistency_detail} imply $\widehat{\bs\Sigma}_\theta\pconv{\bs\Sigma}_\theta$ and $\widehat{\bs\sigma}_{\psi\theta}\pconv{\bs\sigma}_{\psi\theta}$. Given that $\widehat\xi_n\pconv\xi_0$ and $\widehat {\bs{J}}_n\pconv {\bs{J}}_0$, it suffices to apply the continuous mapping theorem to obtain \eqref{eq:Vhat_component}. In particular, expanding the (uncentered) second moment of $\widehat\varphi_t$ gives exactly the three terms in \eqref{eq:Vhat_component}. Moreover,
\[
 \overline{\widehat\varphi} =\widehat\Psi_n^{-1}\frac{1}{n}\sum_{t=1}^n\widehat\psi_t -\widehat\xi_n\widehat J_n'
  (\widehat {\bs{H}}_n^{\,F})^{-1}  \frac{1}{n}\sum_{t=1}^n\widehat {\bs{s}}_t=o_p(1).
\]
The first sample average is zero by the defining expectile equation. The second is zero provided that the interior QML estimator satisfies its first-order condition, otherwise it is of order $o_p(1)$ by the plug-in law of large numbers. This means that subtracting $\overline{\widehat\varphi}^{\,2}$ is asymptotically negligible, thereby establishing the consistency of the alternative estimator $\widehat V_{F,\xi}$.
\end{proof}

\begin{remark}[Fully-factorized covariance estimators]
Alternatively, we can also estimate $\bs\Sigma_\theta$ and $\bs\sigma_{\psi\theta}$ respectively by
\[
  \frac{\sum_{t=1}^n a(\widehat\eta_t)^2}{n\,\widehat\kappa_h^2}\,\widehat{\bs{I}}_n^{-1}\quad\mbox{and}\quad\frac{\sum_{t=1}^n a(\widehat\eta_t)\widehat\psi_t}{n\,\widehat\kappa_h}\,\widehat{\bs{I}}_n^{-1}\widehat{\bs{J}}_n
\]
to exploit the factorization result.
\end{remark}

Finally, it remains to prove that the feasible prediction interval for the conditional expectiles is asymptotically valid.

\begin{proof}[Proof of Corollary \ref{cor:conditional_expectile_interval}]
Assumptions~\ref{ass:approximation} and~\ref{ass:forecast} imply that $\widetilde\sigma_{n+1}(\widehat{\bs\theta}_n)-\sigma_{n+1}(\bs\theta_0)=o_p(1)$ and $\widehat{\bs{D}}_{n+1}-\bs{D}_{n+1}=o_p(1)$ given the consistency of $\widehat{\bs\theta}_n$ and the derivative envelope. Proposition~\ref{prop:variance_estimation} then provides the consistency of the remaining components of the variance estimator, so that $\widehat{V}_{n+1}-V_{n+1}\pconv0$. In addition, positive definiteness of $\bs\Sigma_Y$ and the lower bound on the conditional scale imply that $V_{n+1}$ is bounded away from zero. This means that $\widehat{V}_{n+1}^{[m_n]}/V_{n+1}\pconv1$ since $V_{n+1}^{[m_n]}-V_{n+1}=o_p(1)$. For every $\varepsilon>0$, the conditional probability that this ratio differs from one by more than $\varepsilon$ converges to zero in probability because its expectation coincides with the corresponding unconditional probability. It then suffices to apply the conditional Slutsky's theorem to obtain the result, with the conditional coverage statement ensuing immediately from the continuity of the standard Gaussian distribution at $\pm z_{1-\alpha/2}$.
\end{proof}

\end{document}